\documentclass[pdflatex,sn-basic]{sn-jnl}

\usepackage{graphicx}%
\usepackage{multirow}%
\usepackage{amsmath,amssymb,amsfonts}%
\usepackage{amsthm}%
\usepackage{mathrsfs}%
\usepackage[title]{appendix}%
\usepackage{xcolor}%
\usepackage{textcomp}%
\usepackage{manyfoot}%
\usepackage{booktabs}%
\usepackage{algorithm}%
\usepackage{algorithmicx}%
\usepackage{algpseudocode}%
\usepackage{listings}%
\usepackage[nameinlink,noabbrev]{cleveref}
\crefname{assumption}{assumption}{Assumptions}
\usepackage{comment}

\newcommand{\cX}{\mathcal{X}}
\newcommand{\cN}{\mathcal{N}}
\newcommand{\RR}{\mathbb{R}}
\newcommand{\PP}{\mathbb{P}}
\newcommand{\EE}{\mathbb{E}}
\newcommand{\sign}{\operatorname{sign}}
\newcommand{\eff}{{\rm eff}}

\newtheorem{theorem}{Theorem}

\newtheorem{proposition}{Proposition}
\newtheorem{lemma}{Lemma}
\newtheorem{corollary}{Corollary}

\theoremstyle{thmstyletwo}
\newtheorem{example}{Example}
\newtheorem{remark}{Remark}
\newtheorem{assumption}{Assumption}

\theoremstyle{thmstylethree}

\begin{document}

\title[On randomized step sizes in Metropolis--Hastings algorithms]{On randomized step sizes in Metropolis--Hastings algorithms}


\author*[1]{\fnm{Sebastiano} \sur{Grazzi}}\email{sebastiano.grazzi@unibocconi.it}

\author*[2]{\fnm{Samuel} \sur{Livingstone}}\email{samuel.livingstone@ucl.ac.uk}

\author[3]{\fnm{Lionel} \sur{Riou-Durand}}\email{lionel.riou-durand@insa-rouen.fr}

\affil[1]{\orgdiv{Department of Decision Sciences and BIDSA}, \orgname{ Bocconi University}, \orgaddress{\street{Via
Roentgen 1}, \city{City}, \postcode{20136}, \state{Milan}, \country{Italy}}}

\affil[2]{\orgdiv{Department of Statistical Science}, \orgname{University College}, \orgaddress{\street{Gower Street}, \city{London}, \postcode{WC1E 6BT}, \country{UK}}}

\affil[3]{\orgdiv{Laboratoire de Math\'{e}matiques} \orgname{de l'INSA Rouen Normandie}, \orgaddress{\street{Avenue de l'Universit\'{e}}, \city{76801 Saint-\'{E}tienne-du-Rouvray}, \country{France}}}


\abstract{The performance of Metropolis–Hastings algorithms is highly sensitive to the choice of step size, and miss-specification can lead to severe loss of efficiency. We study algorithms with randomized step sizes, considering both auxiliary-variable and marginalized constructions. We show that algorithms with a randomized step size inherit weak Poincar\'{e} inequalities/spectral gaps from their fixed-step-size counterparts under minimal conditions, and that the marginalized kernel should always be preferred in terms of asymptotic variance to the auxiliary-variable choice if it is implementable.  In addition we show that both types of randomization make an algorithm robust to tuning, meaning that spectral gaps decay polynomially as the step size is increasingly poorly chosen. We further show that step-size randomization often preserves high-dimensional scaling limits and algorithmic complexity, while increasing the optimal acceptance rate for Langevin and Hamiltonian samplers when an Exponential or Uniform distribution is chosen to randomize the step size. Theoretical results are complemented with a numerical study on challenging benchmarks such as Poisson regression, Neal's funnel and the Rosenbrock (banana) distribution.}




\maketitle

\section{Introduction}
Markov chain Monte Carlo (MCMC) methods are the gold standard for asymptotically unbiased Bayesian inference and consist of simulating a discrete-time Markov chain whose ergodic distribution coincides with the given Bayesian posterior $\pi$. Ergodic averages can then be used to estimate expectations with respect to the posterior. Many popular MCMC methods are based on the Metropolis--Hastings algorithm, in which each step of the Markov chain is simulated by proposing a new state according to a candidate Markov kernel and then accepting or rejecting it depending on the ratio of the \emph{target} distribution $\pi$ evaluated at the current and proposed locations.

The candidate kernel is often parametrized by a step size and algorithm performance crucially depends on this choice: when the step size is too small, the process often behaves like a random walk, while a step size that is too large leads to proposals that are rejected with high probability; in both cases the underlying Markov chain mixes slowly and ergodic averages are inefficient. A key contribution for understanding the trade-off between these regimes is given by diffusion/scaling limits, e.g. \citet{Roberts::2001, Roberts::1997,yang2020optimal, Beskos::2013}. Scaling limits facilitate analysis of the asymptotic performance of an algorithm in high dimensions and provide simple guidelines for optimally tuning the step size. They also imply an order of complexity of the algorithm with respect to the number of dimensions of the target distribution \citep{roberts2016complexity}.

\citet{livingstone2022barker} recently showed that the performance of standard gradient-based methods such as the Metropolis-adjusted Langevin algorithm (MALA) and Hamiltonian Monte Carlo (HMC) can deteriorate (exponentially) quickly when the step size is miss-specified. The framework is particularly insightful for explaining practical limitations of MCMC methods for target distributions with a complex geometric structure  \citep[such as Neal's funnel in][]{neal2011mcmc}, for which it is difficult to identify a global step size that allows the Markov chain to efficiently explore the whole space.

Motivated by this, we identify and explore simple modifications of popular algorithms based on randomizing the step size. We show that this simple modification leads to new methods that are provably less affected when the scale of the step size is miss-specified, while retaining  the same complexity when studied in the high-dimensional scaling limit regime. Step size randomization, or \emph{jitter}, has been proposed many times in the literature as a useful heuristic \citep[e.g.][]{neal2011mcmc}.  Here we offer a rigorous analysis, providing clear justifications for both theorists and practitioners.

\subsection{Main contributions}
We identify two approaches to endowing a Metropolis--Hastings algorithm with a randomized step size: an auxiliary-variable approach and a marginalized approach \citep[following the terminology of][]{titsias2018auxiliary}. We highlight that the marginalized approach is always more efficient in terms of asymptotic variance, but show empirically that the auxiliary approach often performs equally well and is a more tractable option. We also illustrate that, under mild conditions, both randomized kernels inherit a weak Poincar\'{e} inequality/spectral gap whenever the base kernel has one.  In both cases we analyze the spectral gaps as the mismatch between the step size and the scale of some of the target components becomes large, and show that the randomized step size algorithms are more robust than their base counterparts. In addition to this, we analyze the optimal choice of $h$ through the lens of scaling limits for randomized step size algorithms. We show that the step size randomization does not affect the high-dimensional complexity properties as measured by expected squared jumping distance, but it does in fact suggest an increased optimal acceptance rate when compared to the base algorithms.  We illustrate this with two concrete examples in the case of MALA and Hamiltonian Monte Carlo (HMC).
We complement our theoretical findings with a comprehensive numerical study, illustrating the benefits of a randomized step size on challenging benchmarks such as Neal's funnel and the Rosenbrock (banana) distribution.  

\subsection{Related literature}
The robustness to tuning framework considered here was introduced in \cite{livingstone2022barker}, where it was used to motivate Metropolis--Hastings algorithms with Barker proposals. In our work, we show that similar robustness properties can also be achieved by randomizing the step sizes of classical algorithms. Scaling limits are a widely used approach in the study of MCMC algorithms, originating in the seminal work \citet{Roberts::1997}.  Our analysis is most closely related to other works directly targeting the expected squared jumping distance \citep[e.g.][]{sherlock2009optimal, vogrinc2023optimal}.

\section{MCMC Algorithms with randomized step size}

\subsection{Metropolis--Hastings algorithms}\label{sec: mh algorithm}
The Metropolis--Hastings algorithm \citep{metropolis1953equation, Hastings::1970} is a generic algorithm used to produce a Markov chain $(X_i)_{i=0,1,\dots}$ with limiting distribution equal to a given target distribution $\pi$ on some space $\cX$. Starting from an initialization $X_0$, each random variable $X_{i+1}|(X_i = x)$ is simulated by first drawing a `proposal' $y$ from a candidate Markov kernel $Q_h(x, \cdot)$, which we index by a tuning parameter $h > 0$ (e.g. a step size). Assuming that $\cX \subset \mathbb{R}^d$, $\pi(dx):= \pi(x)dx$ and $Q_h(x,dy) := Q_h(x,y)dy$ for all $x,y$, then $X_{i+1}|(X_i = x)$ is set to be equal to $y$ with probability 
\begin{equation}
    \label{eq: acceptance rejection}
    \alpha_h(x,y) = \min\left(\frac{\pi(y)Q_h(y, x)}{\pi(x)Q_h(x, y)} ,1\right)
\end{equation}
whenever $\pi(x)Q_h(x,y) >0$, and equal to $x$ otherwise. The corresponding Markov transition kernel takes the form
\begin{equation}
    \label{eq: MH kernel}
    P_h(x,d y) = Q_h(x,d y)\alpha_h(x,y)+\left(1-\int_\mathcal{X} \alpha_h(x,u) Q_h(x,d u)\right)\delta_x(d y).
\end{equation}
It is not difficult to verify that $P_h$ is $\pi$-invariant for any $h > 0$, since it satisfies the detailed balance condition relative to $\pi$, see for example \cite{roberts2004general} for an overview.

\subsection{Randomized step sizes via auxiliary-variable and marginalized kernels}
\label{sec:randomized_mh}
We consider two different approaches to endowing a Metropolis--Hastings kernel $P_h$ with a randomized step size.  In both cases the step size will be generated via a distribution $\mu(dz) := \mu(z)dz$ with positive mass on a subset of $(0,+\infty)$, for which we also assume the following.

\begin{assumption} \label{ass:mu_positive_at_zero}
There exists $C>0$ such that for all $h \ge 1$ and any $z > 0$
\begin{equation}
\mu(z/h)\ge C\mu(z).
\end{equation}
\end{assumption}

Assumption \ref{ass:mu_positive_at_zero} ensures that $\mu$ has non-negligible mass near zero, which is important for the spectral gap robustness results of Section \ref{sec:analysis}.  

Let $(P_{h})_{h>0}$ be the family of  Metropolis--Hastings kernels satisfying \eqref{eq: MH kernel} and indexed by some parameter $h>0$.  
Throughout, we assume that for all $x \in \mathcal{X}$ and any $A$ in the associated Borel $\sigma$-field, the function $z \mapsto P_{hz}(x,A)$ is $\mu$-measurable.

\begin{example}
We define the \emph{Auxiliary-variable} transition kernel as the mixture
$$
\overline{P}_h(x,d y):= \int_0^{\infty} \mu(z)P_{hz}(x,d y)d z.
$$
\end{example}

We can simulate $X_{i+1}$ from $\overline{P}_h(x,\cdot)$ given the current state $X_i = x$ in two simple steps: 
\begin{enumerate}
    \item Sample $z \sim \mu$
    \item Sample $X_{i+1} \sim P_{hz}(x,\cdot)$ as in \eqref{eq: MH kernel}.
\end{enumerate} 
See Algorithm~\ref{alg:aux_mh} for more detail.  

\begin{algorithm}[]
\caption{Inner loop Auxiliary-variable MH}\label{alg:aux_mh}
\begin{algorithmic}[1]
\Require $X_{i} = x \in \cX$, $h > 0$,  $\mu(\cdot)$, family of kernels $(Q_h)_{h>0}$.
\State $z \sim \mu$
\State $y \sim Q_{hz}(x, \cdot)$
\If {$\mathrm{Unif}([0,1]) \le \alpha_{hz}(x,y)$} 
\State $X_{i+1} = y$ 
\Else
\State $X_{i+1} = x$ 
\EndIf
\State \Return $X_{i+1}$.
\end{algorithmic}
\end{algorithm}

\begin{example}
We define the \emph{Marginalized} kernel as
\begin{equation}
    \label{eq: marginalized_kernels}
M_h(x,d y)=\overline{Q}_h(x, y)\widetilde\alpha_h(x,y) d y+\left(1-\int_\mathcal{X} \overline{Q}_h(x,u)\widetilde\alpha_h(x,u) du \right)\delta_x(d y)
\end{equation}
where
$$
\overline{Q}_h(x,y):=\int_0^{\infty} \mu(z)Q_{hz}(x, y)d z
$$
and 
$$
\widetilde\alpha_h(x,y) = \min\left(\frac{\pi( y)\overline{Q}_h(y, x)}{\pi(x)\overline{Q}_h(x,y)}, 1\right)
$$
is chosen such that each member of the family $(M_h)_{h>0}$ is $\pi$-reversible. 
\end{example}

While $M_h$ always exists, it is only possible to simulate from it when a closed-form expression is available for the transition density $\overline Q_h$. The only difference between the auxiliary-variable kernel $\overline{P}_h$ and the marginalized kernel $M_h$ is in the acceptance probability in Algorithm~\ref{alg:aux_mh}, line 3, which changes from $\alpha_{hz}(x,y)$ for the auxiliary kernel to  $\widetilde \alpha_h(x,y)$ for the marginalized kernel, for which the variable $z$ has been marginalized. 

The auxiliary-variable construction for randomized step sizes has been known and advocated for some time, e.g. \cite{neal2011mcmc} suggests this in the context of Hamiltonian Monte Carlo to avoid proposals being either identical or very close to the current state when the Hamiltonian flow is periodic.  To the best of our knowledge, however, there is little rigorous understanding of how this influences the properties of an algorithm.

\begin{remark}
    When the proposal kernel is symmetric, i.e. $Q_h(x,y) = Q_h(y,x)$, then $\alpha_{hz} = \alpha_{z}  = \tilde \alpha_h$ and Marginalized and Auxiliary-variable kernels coincide. 
\end{remark}

\subsection{Examples of randomized kernels}
\label{subsec:examples_randomized_kernels}

As shown above, sampling from the auxiliary-variable kernel is a straightforward procedure, but this is not necessarily true for the marginalized kernel $M_h$.  Both the Gaussian random walk Metropolis and the Barker proposal of \cite{livingstone2022barker} can, however, be derived for example as marginalized versions of simpler base algorithms when the dimension $d = 1$.  In the case of the random walk, consider the Metropolis--Hastings algorithm with proposals generated from the current state $x$ via the transformation $y = x + \sqrt{h}\xi$, where $\xi \in \{-1,+1\}$ and $\mathbb{P}(\xi = -1) = 1/2$. Choosing the randomizing distribution $\mu$ to be half-Normal, meaning a standard Normal distribution restricted to the positive half-line, means that the marginalized kernel $M_h$ becomes a random walk Metropolis with Gaussian noise of variance $h$.  Similarly, a \emph{Rademacher Barker proposal} algorithm has been discussed in \cite{vogrinc2023optimal}, in which proposals take the form $y = x + b(x,\sqrt{h}\xi)\cdot \sqrt{h}\xi$, with $\xi$ following the same Rademacher distribution and $b(x,\sqrt{h}\xi) \in \{-1, +1\}$ such that $\mathbb{P}(b(x,\sqrt{h}\xi) = +1) := 1/(1+\exp(- (\log\pi)'(x)\cdot \sqrt{h}\xi))$.  Again choosing $\mu$ to be half-Normal induces the standard Barker proposal scheme of \cite{livingstone2022barker}, and setting $\mu$ to be a Gaussian truncated to be positive with mean $\surd(1-\sigma^2)$ and variance $\sigma^2$ for some $\sigma < 1$ gives the bi-modal Barker proposal described in Section 4.1 of \citet{vogrinc2023optimal}.

We introduce here two new randomized kernels in the multi-dimensional case, corresponding to marginalized versions of the Metropolis-adjusted Langevin algorithm (MALA), for which the base proposal kernel takes the form
\begin{equation}
    \label{eq: mala kernel}
    Q_h(x, y) = \phi_d(y ; \, x + h\nabla \log\pi(x), 2h I_{d\times d})
\end{equation}
where $\phi_d(y ; \, \mu, \Sigma)$ denotes the  $d$-dimensional Gaussian density with mean $\mu$ and covariance matrix $\Sigma$ evaluated at $y$. The proposal in \eqref{eq: mala kernel} is justified as it corresponds to an increment of the Euler discretization with step size $h$  of the Langevin diffusion
$$
d X_t = \nabla \log \pi(X_t)  d t + \sqrt{2} d B_t
,$$
which is known to be $\pi$-invariant, see \cite{RobertsTweedie::1996} for details.  We identify and analyze two families of implementable Marginalized kernels \eqref{eq: marginalized_kernels}, which are based on randomizing the step size of the MALA kernel in \eqref{eq: mala kernel}. The first modification is when $\mu$ is a uniform distribution on $[0,1]$. In this case, 
\begin{align}
\label{eq: bessel}
\overline{Q}_h(x,y) \propto e^{c}  \check K_{1-d/2}(a, b),
\end{align}
where $a = \|y - x\|^2 /(4h)$, $b = h \|\nabla \log \pi(x)\|^2/4$, $c = \langle y - x, \, \nabla \log \pi (x) \rangle /2$ and $\check K_\nu(a, b) = \int_1^\infty t^{-\nu-1} \exp(-at - b/t) dt$ is the upper incomplete modified Bessel function of the second kind. When the dimension $d$ is odd, \eqref{eq: bessel} can be re-expressed in terms of the complementary error function.

The second version is when $\mu$ follows an $\text{Exp}(1)$ distribution, leading to 
$$
\overline{Q}_h(x,y) \propto  e^{c} \left(\frac{a}{b}\right)^{\frac{1}{2} - \frac{d}{4}} K_{1 - d/2}(2 \sqrt{ab}),
$$
where $a = \|y - x\|^2 /(4h)$, $b = 1 + h \|\nabla \log \pi(x)\|^2/4$, $c = \langle y - x, \, \nabla \log \pi(x) \rangle /2$ and $K_\nu(x)$ is the modified Bessel function of the second kind, which is readily computable by standard scientific programming languages. More generally, whenever $\mu$ is chosen as an instance of the Generalized Inverse Gaussian (GIG) family, then the marginalized kernel will follow a generalized hyperbolic distribution, with density expressed in terms of a modified Bessel function as above \citep{jorgensen2012statistical}.  Other natural choices for $\mu$ within the GIG family are Gamma, inverse Gamma and inverse Gaussian distributions.  We compare both of these new marginalized algorithms to their auxiliary and base counterparts in Section \ref{sec:experiments}.

\section{Spectral gap analysis} \label{sec:analysis}
We present here several results related to Auxiliary-variable and Marginalized Metropolis--Hastings transition kernels.  We first consider Dirichlet forms, comparing the kernels to each other and $P_h$, and also showing in Section \ref{subsubsec:robustness_analysis} that for any base kernel $P_h$ both $\overline P_h$ and $M_h$ will attain the `robustness to tuning' property that spectral gaps decay at a polynomial rate in $h^{-1}$ as $h \uparrow \infty$ \citep{livingstone2022barker}. Crucially, this robustness analysis requires no assumptions on the target distribution other than that $P_1$ possesses a spectral gap.  In the absence of this, our conclusions could also be modified to the case of weak Poincar\'{e} inequalities using the results of \cite{andrieu2022comparison}.  We also note that the results hold in general for any Markov kernel indexed by a one-dimensional parameter $h$, and therefore extend to other cases than when $h$ is the algorithmic step size.  

\subsection{Dirichlet forms, Poincar\'{e} inequalities \& spectral gaps}

For any $f:\mathcal{X} \to \mathbb{R}$, and any $\pi$-reversible Markov kernel $P$, we define an associated Dirichlet form
$$
\mathcal{E}(P,f) := \frac{1}{2}\int (f(y)-f(x))^2\pi(d x) P(x,d y).
$$
Let $L_0^2(\pi) := \{ f:\mathcal{X} \to \mathbb{R} ~|~ \int f(x)^2 \pi(d x) <\infty,\, \int f(x) \pi(dx) = 0 \}$.  The Markov chain $(X_i)_{i=1,2,\dots}$ with transition kernel $P$ satisfies a \emph{Poincar\'{e} inequality} with Poincar\'{e} constant $c \in (0,\infty)$ if for all $f \in L_0^2(\pi)$ 
\begin{equation} \label{eq:poincare}
\text{Var}_\pi(f) \leq c \mathcal{E}(P,f).
\end{equation}
The (right) \emph{spectral gap} for $P$ is then defined as
\[
\mathrm{Gap}(P) = \inf_{f \in L^2_0(\pi)} \frac{\mathcal{E}(P,f)}{\text{Var}_\pi(f)}.
\]
 If \eqref{eq:poincare} holds with $c>0$, then $\mathrm{Gap}(P)>0$ and it is straightforward to see that $1/\mathrm{Gap}(P)$ is the smallest choice of $c$ for which \eqref{eq:poincare} can hold. For any kernel $P$ the corresponding Markov operator $P$, defined point-wise as $Pf(x) := \int f(y)P(x,dy)$ for any $f \in L_0^2(\pi)$, is called positive if it has a spectrum that is wholly contained in $[0,1]$. If $P$ is a positive operator and \eqref{eq:poincare} holds, then for any $f \in L^2_0(\pi)$
$$
\text{Var}_\pi(Pf) \leq (1-\text{Gap}(P))^2 \cdot \text{Var}_\pi(f).
$$
From this equation, geometric convergence in total variation distance towards equilibrium can be deduced.  Precisely, noting for any signed measure $\nu$ that $\|\nu\|_{TV} := \sup_{f :\mathcal{X} \to [0,1]}\left|\nu(f)\right|$, it holds that
\[
\|\nu P^n - \pi\|_{TV} \leq \frac{1}{2}\chi^2(\nu||\pi)^{1/2}\cdot (1-\text{Gap}(P))^n,
\]
where $\nu P(\cdot) := \int \nu(dx)P(x,\cdot)$ and $\chi^2(\nu||\pi) := \int ((d\nu/d\pi)(x) - 1)^2\pi(dx)$ denotes the $\chi^2$-divergence between $\nu$ and $\pi$.  In words, if the chain is initialized from a probability distribution $\nu$, then the distance to equilibrium decays geometrically with rate $(1-\text{Gap}(P))$, and the distance also depends on how close $\nu$ is to the limiting distribution $\pi$.  When $P$ is not positive then it can be replaced with its lazy counterpart $(P+I)/2$.

The spectral gap can also be related to the asymptotic variance of ergodic averages.  Indeed if $(X_i)_{i=1,2,\dots}$ is a Markov chain with transition kernel $P$ and $X_0 \sim \pi$ then for any $g \in L^2_0(\pi)$ it holds that
\[
\lim_{n \to \infty} \frac{1}{n}\text{Var}\left(\sum_{i=1}^n g(X_i) \right) \leq \left( \frac{2}{\text{Gap}(P)} - 1 \right)\text{Var}_\pi(g).
\]
This can be seen by first noting that if $\text{Gap}(P) > 0$ then the asymptotic variance for $g$ can be written $2\mathcal{E}(P,h_g) - \text{Var}_\pi(g)$, where $h_g$ solves the Poisson equation $(I-P)h_g = g$ \citep[e.g.][]{roberts2004general}, and then noting that in this case the Dirichlet form can be straightforwardly re-expressed as $\mathcal{E}(P,h_g) = \int g(x)h_g(x)\pi(dx)$.  The Cauchy--Schwarz inequality therefore implies that $\mathcal{E}(P,h_g) \leq \surd\text{Var}_\pi(g)\surd\text{Var}_\pi(h_g)$, and applying the Poincar\'{e} inequality \eqref{eq:poincare} to $\surd \text{Var}_\pi(h_g)$ with the constant $c = 1/\text{Gap}(P)$ gives upon rearrangement the inequality $\mathcal{E}(P,h_g) \leq \text{Var}_\pi(g)/\text{Gap}(P)$. Substituting into the asymptotic variance expression gives the result.

In many settings a Markov chain does not have a positive spectral gap, or equivalently a Poincar\'{e} inequality does not hold.  In this case Dirichlet forms can still be used to understand the properties of the Markov chain via a \emph{weak Poincar\'{e} inequality}
\begin{equation} \label{eq:wpi}
\text{Var}_\pi(f) \leq s\mathcal{E}(P,f) + \beta(s)\Phi(f),
\end{equation}
for some $s > 0$ and suitable real-valued functions $\beta$ and $\Phi$  \citep[which must also satisfy additional conditions, see e.g. Section 2.2.1 of][]{andrieu2022comparison}.  It can be shown that if \eqref{eq:wpi} holds then the Markov chain converges to equilibrium at a rate that is connected to the functions $\beta$ and $\Phi$ \citep[see e.g.][]{andrieu2022comparison}.

\subsection{Orderings and Robustness to tuning}
Markov chains can be compared in various ways. The approach introduced by Peskun \citep{peskun1973optimum} is to construct a partial order on $\pi$-reversible Markov transition kernels. If $P_1$ and $P_2$ are two such kernels then $P_1 \succeq P_2$ in the sense of Peskun if $P_1(x,A\cap \{x\}^c) \geq P_2(x,A\cap \{x\}^c)$ for all $x \in \mathcal{X}$ and any event $A$.  Peskun showed that when $\mathcal{X}$ is finite such an ordering implies that asymptotic variances computed using ergodic averages from $P_1$ will be smaller than those computed from $P_2$ for any square-integrable test function $g$.  Tierney \citep{tierney1998note} then extended this result to general state spaces.  It is not difficult to see that if $P_1 \succeq P_2$ then $\mathcal{E}(P_1,f) \geq \mathcal{E}(P_2,f)$ for all $f \in L_0^2(\pi)$, from which it is natural to consider directly ordering Dirichlet forms.  The same asymptotic variance conclusion holds under this ordering owing to its Dirichlet form representation given in the previous section.

The criterion of Peskun, often called `off-diagonal domination', can be too stringent to be satisfied in many applications, and applies only to reversible Markov chains.  It has been extended in various directions, see e.g. \cite{mira1999ordering,andrieu2021peskun}. One useful generalization for our purposes is to consider for some $\omega > 0$ the relaxation
\begin{equation} \label{eq:weak_peskun}
\mathcal{E}(P_1,g) \geq \omega \mathcal{E}(P_2,g).
\end{equation}
This weakened ordering implies that $\text{Gap}(P_1) \geq \omega \text{Gap}(P_2)$, and also that asymptotic variances under $P_1$ are less than those under $P_2$ up to a constant additive term depending on $\omega$ \citep{andrieu2018uniform}.  In addition \cite{andrieu2022comparison} show that in the absence of a spectral gap conditions such as \eqref{eq:weak_peskun} can be used to establish that a weak Poincar\'{e} inequality holds for $P_1$ provided that such an inequality holds for $P_2$.

The ordering \eqref{eq:weak_peskun} was used in \cite{livingstone2022barker} to assess the robustness of a Metropolis--Hastings algorithm to the choice of step size.  The authors introduce a framework in which the algorithm targets a distribution with density $\pi^{(\lambda)}$, where $\lambda$ denotes the scale of the first component.  Crucially, the same choice of algorithmic step size is chosen for all different values of $\lambda$, meaning that as $\lambda \downarrow 0$ the result is an algorithm for which the step size is much larger than the optimal choice, and when $\lambda \uparrow \infty$ it is much smaller than optimal.  The authors study the rate at which the spectral gap of the corresponding Markov chain decays to zero as either $\lambda \downarrow 0$ or $\lambda \uparrow \infty$.  If the rate is polynomial in $1/\lambda$ as $\lambda \downarrow 0$ then the algorithm is called \textit{robust to tuning}. The authors show that the random walk Metropolis and Barker proposal are robust to tuning as $\lambda \downarrow 0$ while MALA and HMC are not, while all algorithms behave similarly as $\lambda \uparrow \infty$.  They then showcase through numerical experiments that robustness to tuning is a very beneficial property for learning optimal algorithmic tuning parameters using adaptive Markov chain Monte Carlo.

It can be shown that studying the above algorithms with a fixed choice of step size as the target $\pi^{(\lambda)}$ changes in the manner described above is equivalent to studying the algorithms with a fixed target distribution $\pi^{(1)}$ where the algorithmic step size $h$ is allowed to vary.  Precisely, there is an isomorphism between the Metropolis--Hastings algorithm with proposal kernel $Q_1$ targeting $\pi^{(\lambda)}$ and the algorithm with proposal kernel $Q_h$ targeting $\pi^{(1)}$ if $h := 1/\lambda^2$, which implies that the two resulting Markov chains will have the same spectral gap \cite[see the supplement to][for details]{livingstone2022barker}.  Re-framing the results of \cite{livingstone2022barker} in these terms, the implication is that for MALA and HMC algorithm performance decays extremely quickly as the step size is increased beyond the optimal choice, meaning that performance is crucially dependent on this parameter.  By contrast, in the random walk Metropolis and Barker proposal the decay in efficacy is much less severe.  The authors showcase that in the context of learning a step size iteratively through a stochastic approximation scheme depending on the Markov chain samples, as is commonly done in practice, if the sample quality is still adequate when the step size is larger than optimal then this can facilitate much faster optimization.

MALA and the Barker proposal \citep{livingstone2022barker,hird2020fresh,vogrinc2023optimal} are natural comparators using this framework since the step sizes in both algorithms clearly have the same physical meaning.  Both are Metropolis--Hastings algorithms in which the proposal can be constructed by numerically integrating the overdamped Langevin diffusion for a fixed time-step, and then applying a Metropolis--Hastings filter to correct for discretisation bias at equilibrium.  In the case of the Langevin algorithm this fact is well-known and the numerical scheme is the classical Euler--Maruyama method. In the Barker case the numerical scheme is the skew-symmetric method introduced in \cite{iguchi2024skew}. In both cases the step size $h$ corresponds to the physical time for which the diffusion process is integrated forward to generate a proposal. What is clear from the results of \cite{livingstone2022barker} is that the rate at which the numerical scheme becomes unstable is different between the two algorithms.

\subsection{Spectral gap results for kernels with randomized step size}
\label{subsubsec:robustness_analysis}
We now present ordering and robustness results for Metropolis--Hastings algorithms with a randomized step size, considering both the auxiliary-variable and marginalized approaches introduced in Section \ref{sec:randomized_mh}.  We begin with a direct ordering.

\begin{proposition}\label{prop:ordering_margvsaux} For any $f \in L^2_0(\pi)$, any $h>0$ and any base kernel $P_h$ we have
\[
\mathcal{E}(M_h,f) \ge \mathcal{E}(\overline P_h, f),
\]    
from which it follows that
\[
\mathrm{Gap}(M_h) \ge \mathrm{Gap}(\overline P_h),
\]
and that an ergodic average constructed using $M_h$ will always have smaller asymptotic variance than the same ergodic average using $\overline P_h$.
\end{proposition}

This result follows directly from Proposition 1 of \cite{titsias2018auxiliary}, see also \cite{storvik2011flexibility}.  The practical conclusion is that if the marginalized kernel is tractable then this should always be preferred.  We show with the next proposition and in the numerical simulations in Section \ref{sec:experiments}, however, that the difference in performance can be very small in many practical cases of interest.

\begin{proposition}[Reverse comparison for randomized Metropolis--Hastings kernels]
For $z>0$ write
\[
        r_{hz}(x,y)
        :=
        \frac{\pi(y)Q_{hz}(y,x)}
             {\pi(x)Q_{hz}(x,y)}
\]
whenever the denominator is positive. Suppose that there exists $K\geq 1$
such that for $\pi(dx)\overline Q_h(x,y)dy$-almost every $(x,y)$
\[
        \operatorname*{ess\,sup}_{z,z'}
        \frac{r_{hz}(x,y)}{r_{hz'}(x,y)}
        \le K .
        \tag{11}
\]
Then
\[
        \operatorname{Gap}(\overline P_h)
        \ge
        K^{-1}\operatorname{Gap}(M_h),
\]
and for any $g\in L^2_0(\pi)$, writing
\[
\sigma^2(P, g) := \lim_{n \to \infty} \frac{1}{n}\textup{Var}\left(\sum_{i=1}^n g(X_i) \right), \quad X_i \sim P(X_{i-1},\cdot)
\]
with $X_0 \sim \pi$, it holds that
\[
        \sigma^2(\overline P_h,g)
        \le
        K\sigma^2(M_h,g)
        +(K-1)\operatorname{Var}_\pi(g)
        \tag{13}
\]
whenever these quantities are finite.
\end{proposition}

\begin{proof}
For $x\neq y$, write
\[
        r_{hz}(x,y)
        =
        \frac{\pi(y)Q_{hz}(y,x)}
             {\pi(x)Q_{hz}(x,y)}
\]
and let
\[
        \nu_{x,y}(dz)
        =
        \frac{Q_{hz}(x,y)}{\overline Q_h(x,y)}\,\mu(dz),
\]
with the case $\overline Q_h(x,y)=0$ being trivial. The off-diagonal
density of $\overline P_h$ is
\[
        \overline Q_h(x,y)
        \int \min\left(1, r_{hz}(x,y)\right)\nu_{x,y}(dz),
\]
whereas the off-diagonal density of $M_h$ is
\[
        \overline Q_h(x,y)
        \min\left(
        1,
        \int r_{hz}(x,y)\,\nu_{x,y}(dz)
        \right).
\]
By assumption
\[
        r_{hz}(x,y)
        \ge
        K^{-1}\int r_{hu}(x,y)\,\nu_{x,y}(du),
        \qquad \nu_{x,y}\text{-a.s.},
\]
meaning
\[
        \int \min\left(1, r_{hz}(x,y)\right)\nu_{x,y}(dz)
        \ge
        K^{-1}\min\left(
        1,
        \int r_{hz}(x,y)\,\nu_{x,y}(dz)
        \right).
\]
This implies that for any $f \in L_0^2(\pi)$ $\mathcal{E}(\overline P_h,f)\ge K^{-1}\mathcal{E}(M_h, f)$.  The results then follow directly by applying Lemma 32 of \cite{andrieu2018uniform}.
\end{proof}

Using the above, we can compare both $M_h$ and $\overline P_h$ to the base kernel $P_h$ in terms of spectral gaps.  We begin with the following assumption on $P_h$.

\begin{assumption} \label{ass:p_has_a_spectral_gap}
The kernel $P_h$ satisfies one of the following conditions:
\begin{itemize}
    \item[(i)] (Weak Poincar\'{e} inequality) There is a $\mu$-measurable $H \subset (0,\infty)$ such that for any $h \in H$ the weak Poincar\'{e} inequality
    \[
    \textup{Var}_\pi(f) \leq c(h) \mathcal{E}(P_h,f) + \beta(c(h))\Phi(f)
    \]
    holds with $0 < c(h) < \infty$, for suitable functions $\beta$ and $\Phi$.
    \item[(ii)] (Poincar\'{e} inequality) Part (i) holds with $\beta(c(h)) \equiv 0$.
\end{itemize}
\end{assumption}

Assumption \ref{ass:p_has_a_spectral_gap}(ii) implies that for any $h \in H$ the kernel $P_h$ has a positive spectral gap, and that
\[
\text{Gap}(P_h) \geq \frac{1}{c(h)}.
\]
Provided that either this or Assumption \ref{ass:p_has_a_spectral_gap}(i) is true then some conclusions can be drawn about $\overline P_h$.  We begin with a technical result.

\begin{lemma} \label{lem:dirichlet_form_aux}
    For any $f \in L^2_0(\pi)$, it holds that
    \begin{equation}
    \mathcal{E}(\overline P_h,f) = \int \mathcal{E}(P_{hz},f)\mu(dz).
    \end{equation}
\end{lemma}

\begin{proof}
Direct calculation gives

\begin{align*}
2\mathcal{E}(\overline P_h,f) 
&= 
\int (f(y)-f(x))^2\pi(dx)\overline P_h(x,dy)
\\
&=
\int \int (f(y)-f(x))^2\pi(dx)P_{hz}(x,dy)\mu(dz)
\\
&=
\int 2\mathcal{E}(P_{hz},f)\mu(dz)
\end{align*}
as required.
\end{proof}

We can combine Lemma \ref{lem:dirichlet_form_aux} with Assumption \ref{ass:p_has_a_spectral_gap}(i) to conclude a weak Poincar\'{e} inequality and Assumption \ref{ass:p_has_a_spectral_gap}(ii) to conclude a positive spectral gap for $\overline P_h$ provided that $\mu$ gives positive mass to the set $H/h := \{ z \in \mathbb{R}: hz \in H\}$, as stated below.

\begin{proposition} \label{prop:spectral_gap_random_vs_base}
    If Assumption \ref{ass:p_has_a_spectral_gap}(i) holds and $\mu(H/h) > 0$ then $\overline P_h$ satisfies a weak Poincar\'{e} inequality of the form
    \[
    \textup{Var}_\pi(f) \leq c_*(h,\mu)\mathcal{E}(\overline P_h, f) + \beta_*(h,\mu)\Phi(f)
    \]
    for any $f \in L^2_0(\pi)$, where
    \[
    c_*(h,\mu)^{-1} := \int_{H/h} \frac{1}{c(hz)}\mu(dz), \quad
    \beta_*(h,\mu) := c_*(h,\mu) \cdot \int_{H/h} \frac{\beta(c(hz))}{c(hz)}\mu(dz).
    \]
    In the case of Assumption \ref{ass:p_has_a_spectral_gap}(ii) then combining with Proposition \ref{prop:ordering_margvsaux} implies
    \[
    \textup{Gap}(M_h) \geq \textup{Gap}(\overline P_h) \geq \frac{1}{c_*(h,\mu)}.
    \]
\end{proposition}

\begin{proof}
The weak Poincar\'{e} inequality can be re-written
\[
\mathcal{E}(P_h,f) \geq c(h)^{-1}\left[ \text{Var}_\pi(f) - \beta(c(h))\Phi(f)\right].
\]
Combining Lemma \ref{lem:dirichlet_form_aux} and Assumption \ref{ass:p_has_a_spectral_gap}(i) therefore gives for any $f \in L_0^2(\pi)$
\begin{align*}
\mathcal{E}(\overline P_{h}, f)
&\geq
\int_{H/h} \mathcal{E}(P_{hz}, f) \mu(dz)
\\
&\geq
\textup{Var}_\pi(f)\int_{H/h}\frac{1}{c(hz)}\mu(dz) - \Phi(f)\int_{H/h} \frac{\beta(c(hz))}{c(hz)}\mu(dz),
\end{align*}
from which upon re-arrangement the first part of the result follows.  Setting $\Phi(f) \equiv 0$ as in Assumption \ref{ass:p_has_a_spectral_gap}(ii) then implies $\textup{Gap}(\overline P_h) \geq 1/c_*(h,\mu)$, from which the second part follows.
\end{proof}

The following immediate corollary highlights the additional robustness that can be offered when randomizing the step size from a suitable $\mu$ using either the marginal or auxiliary-variable approach.
\begin{corollary}\label{cor:cor_spectral_gap_bound}
    If Assumption \ref{ass:p_has_a_spectral_gap}(ii) holds for some $H$ of positive Lebesgue measure on $\mathbb{R}^+$, and the density $\mu(z) > 0$ for all $z >0$, then for any $h >0$,
    \[
    \textup{Gap}(M_h) \geq \textup{Gap}(\overline P_h) > 0.
    \]
\end{corollary}
\begin{example}
    When $\pi$ is Gaussian it is well-known that MALA and Hamiltonian Monte Carlo can only be shown to have a positive spectral gap if $h$ is chosen to be suitably small \citep{RobertsTweedie::1996, livingstone2019geometric, durmus2020irreducibility}. By contrast both $M_h$ and $\overline P_h$ will have a positive spectral gap for any choice of $h$ provided that the randomizing distribution $\mu$ has positive density in a neighbourhood of zero.
\end{example}

\begin{remark}[Non-asymptotic $L^p$ error bounds]
Proposition~\ref{prop:spectral_gap_random_vs_base} can be combined with results in \citet{rudolf2011explicit, hofstadler2025optimal} to derive explicit non-asymptotic $L^p$ error bounds for both Auxiliary-variable and Marginalized Markov chains. In particular, suppose that Assumption~\ref{ass:p_has_a_spectral_gap} (ii) holds, then $c^{-1}_*(h,\mu)$ is a lower bound for both $\mathrm{Gap}(\bar P_h)$ and $\mathrm{Gap}(\bar M_h)$ and the $L^{p}$ error bounds in  \citet[][Theorem 2.1]{hofstadler2025optimal} and \citet[][Theorem 3.41]{rudolf2011explicit} can be directly applied.
\end{remark}

Proposition \ref{prop:spectral_gap_random_vs_base} already illustrates how the spectral gaps of $M_h$ and $\overline P_h$ can be more stable as $h\uparrow \infty$ than that of $P_h$.  To see this, suppose that Assumption \ref{ass:p_has_a_spectral_gap}(ii) is satisfied with $c(h) := e^{h}$, meaning the spectral gap of $P_h$ decays to zero at an exponential rate as $h$ increases.  Choosing $\mu$ to be a half-Normal distribution with $\sigma^2 = 1$, then $c_*(h,\mu)^{-1} = 2e^{h^2/2}(1-\Phi(h))$, where $\Phi$ is the standard Normal cumulative distribution function.  For large $h$ then $c_*(h,\mu) \sim h \sqrt{\pi/2}$, meaning that the spectral gaps of $M_h$ and $\overline P_h$ decay at a slower polynomial rate as $h \uparrow \infty$.  In the next result we show that this robustness to tuning holds more generally for the randomized kernels.  We show below that for any Metropolis--Hastings algorithm $P_h$ indexed by some one-dimensional parameter $h$, the act of randomizing $h$ by using either $M_h$ or $\overline P_h$ imposes a level of robustness, even if this is not the case for the original kernel $P_h$.  We begin with a weak ordering of the Dirichlet forms.

\begin{proposition} \label{prop:dirichlet_aux}
For any $f \in L^2_0(\pi)$ and any $h > 1$, suppose $\mu$ satisfies Assumption \ref{ass:mu_positive_at_zero}, then 
\[
\mathcal{E}(\overline P_h, f) \geq w(h)\mathcal{E}(\overline P_1,f), \quad w(h)^{-1} = O(h).
\]
\end{proposition}

\begin{proof}
Direct calculation and using Assumption \ref{ass:mu_positive_at_zero} gives
\begin{align*}
2\mathcal{E}(\overline{P}_{h},f) 
&=
\int (f(x)-f(y))^2\int_0^{\infty} \mu(z)P_{hz}(x,d y) d z \, \pi(d x)
\\
&=
\int (f(x)-f(y))^2  \int_0^{\infty} \mu\left(\frac{u}{h}\right)P_{u}(x,d y)\frac{1}{h} d u \, \pi(d x)
\\
&\ge
\int (f(x)-f(y))^2  \int_0^{\infty} C\mu(u)P_{u}(x, d y)\frac{1}{h} d u \,\pi(d x)
\\
&=
2\mathcal{E}(\overline{P}_1, f)\times \frac{C}{h}.
\end{align*}

Setting $w(h) := C/h$ therefore establishes the result.
\end{proof}

Combining Proposition~\ref{prop:dirichlet_aux} with Proposition~\ref{prop:ordering_margvsaux} leads to the following, which is the main result of this section.

\begin{theorem} \label{thm: spectral gap decay}
If $\textup{Gap}(\overline P_1) > 0$ and $\mu$ satisfies Assumption \ref{ass:mu_positive_at_zero} then
\[
\textup{Gap}(M_h) = \Omega(h^{-1}), \quad \textup{Gap}(\overline P_h) = \Omega(h^{-1}),
\]
where $f = \Omega(g) \iff \lim \inf_{t\to\infty}f(t)/g(t) >0$ for $f,g:[0,\infty) \to [0,\infty)$.
\end{theorem}

\begin{proof}
From Propositioin~\ref{prop:dirichlet_aux} and the assumption that $\textup{Gap}(\overline P_1) > 0$ the result follows for $\overline P_h$. Given that $\textup{Gap}(M_h) \geq \textup{Gap}(\overline P_h)$ by Proposition~\ref{prop:ordering_margvsaux} then this also implies the result for $M_h$.
\end{proof}

\section{Scaling limits}
\label{subsec:scaling_limits}

Scaling limits are classical tools used to understand how the step size $h$ must be chosen as the dimension $d$ of the target distribution $\pi$ increases to prevent the acceptance rate from degenerating. These results 
give the algorithmic dependence on $d$ and can be used to optimize the average acceptance rate of the algorithm, providing practical guidance for tuning $h$ in high dimensions. 

The general framework makes some regularity assumptions on a sequence of target distributions $\pi_d$ on $\RR^d$  (e.g. being of the product form, although generalizations are possible). Under these assumptions, optimal scaling results show that there exists a constant $c > 0$ such that for all $\ell > 0$, $X \sim \pi_d$, $Y \sim Q_h(X, \cdot)$ and $h_d = \ell d^{-c}$, as $d \to \infty$,
\begin{equation}
    \label{eq:alpha scaling limits}
    a(d, P_{h_d}) := \EE(\alpha(X,Y)) \to a(\ell) := 2 \, \Phi\left(- \kappa \, \ell^{1/(2c)}/2\right)
\end{equation}
where $\kappa>0$ is a constant whose expression depends on the algorithm. The limit in \eqref{eq:alpha scaling limits} identifies the appropriate scaling of the step size $h$ such that the average acceptance rate (i.e. the probability that  the Markov chain changes its value in one step) has a  non-trivial limit in $[0,1]$.  Furthermore, for any two increments of the Markov chain $X_i \sim \pi_d$, $X_{i+1} \sim P_{h_d}(X_i, \,\cdot)$ and the same constants $c, h_d$ as above, as $d \to \infty$, 
\begin{equation}
    \label{eq:ESJD}
 \mathrm{ESJD}(d, P_{h_d}) := d^{c} \, \EE\left[\left(X^{(1)}_i - X^{(1)}_{i+1}\right)^2\right] \to {\rm eff}(\ell) 
 :=  \ell \, a(\ell)
\end{equation}
 where $X^{(j)}$ is the $j$th element of $X$; see \cite{Roberts::2001} for an overview. The quantity on the left hand side of \eqref{eq:ESJD} is the (scaled) expected squared jumping distance (ESJD) of the first coordinate. Maximizing ESJD corresponds to minimizing the 1-lag autocorrelation function of the Markov chain $(X^{(1)}_i)_{i=1,2,\dots},$ and to reducing the asymptotic variance of ergodic  averages in the high-dimensional limit,  see for example \cite{rosenthal2003asymptotic}. Furthermore, the right-hand side of \eqref{eq:ESJD} also corresponds to the speed (or diffusivity) of the Langevin diffusion process obtained as a high-dimensional limit of the (appropriately rescaled) original Markov chain, see \citet{Roberts::1997, Roberts::1998} for details.

If the limits in \eqref{eq:alpha scaling limits} and \eqref{eq:ESJD} hold,  one can optimize the function ${\rm eff}(\ell)$  and find the optimal asymptotic acceptance rate as follows. Letting $u=\ell^{1/(2c)} \kappa/2$, the optimal $\ell$ maximizing the asymptotic ESJD satisfies
\begin{align}
    \frac{\partial}{\partial \ell} \eff(\ell)=0 &\Leftrightarrow2 \Phi(-\ell^{1/(2c)}\kappa/2)- \frac{\kappa}{2c}  \ell^{1/(2c)} \phi(-\ell^{1/(2c)}\kappa/2)=0 \nonumber\\
    &\Leftrightarrow \Phi(-u) - \frac{1}{2c} u \phi(-u)=0 \label{eq:opt_rate_2}
\end{align}
where $\phi$ denotes the standard Gaussian density in one dimension. One can then find numerically $\ell_{\rm opt}$ by solving \eqref{eq:opt_rate_2} and compute the optimal asymptotic acceptance rate $a(\ell_{\rm opt})$.

Classical scaling limits have been established for random walk Metropolis with $c = 1 $ and $a(\ell_{\rm opt}) = 0.234$ \citep{Roberts::1997}, Metropolis-adjusted Langevin algorithm (MALA) and the Barker proposal with $c = 1/3$ and $a(\ell_{\rm opt}) = 0.574$ \citep{Roberts::1998, vogrinc2023optimal} and for HMC and Metropolis-adjusted Langevin trajectories with $c = 1/4$ and $a(\ell_{\rm opt}) = 0.651$ \citep{Beskos::2013, riou2022metropolis}.

Turning now  to our randomized algorithms. Our main result below shows that if a scaling limit holds true for a Markov kernel $P_h$, then similar limits hold for our Auxiliary-variable Markov kernel $\overline P_h$ in Algorithm~\ref{alg:aux_mh}.

\begin{theorem}[Scaling limits for Auxiliary-variable kernels] \label{thm:scaling_auxvar kernles} Suppose that for some sequence of kernels $P_{h_d}$, the limits \eqref{eq:alpha scaling limits}  and \eqref{eq:ESJD}  hold with $h_d = \ell d^{-c},$ for some $c >0$. Then, as $d \to \infty$, 
\[
a(d, \overline P_{h_d})\to \overline a(\ell):= \int_0^\infty a(\ell z)\mu(z)d z.
\]
Moreover, setting $f_d(z) := \mathrm{ESJD}(d, P_{h_d z})$, assuming that the sequence $f_1(z),f_2(z),...$ is uniformly $\mu$-integrable, then
$$
\mathrm{ESJD}(d, \overline P_{h_d}) \to \overline{\rm eff}(\ell):= \int_0^\infty {\rm eff}(\ell z)\mu(z)d z.
$$
\end{theorem}
\begin{proof}
Let $X \sim \pi_d$, $Y \sim Q_{h_d z}(X, \cdot)$. We can use previous optimal scaling results in \eqref{eq:alpha scaling limits}-\eqref{eq:ESJD}, to claim that for any $z>0$, as $d\rightarrow\infty$, we have the point-wise limits
$$
\EE (\alpha(X, Y)) \rightarrow a(\ell z), \qquad \mathrm{ESJD}(d,P_{h_d z}) \rightarrow  {\rm eff}(\ell z).
$$ 
By taking expectations with respect to $z$ the first part of the theorem is proven by applying the bounded convergence theorem.  For the second part note that $\mathrm{ESJD}(d,\overline P_{h_d}) = \int \mathrm{ESJD}(d,P_{h_d z})\mu(dz) = \int f_d(z)\mu(dz)$, so uniform $\mu$-integrability of the sequence $f_1, f_2, ...$ implies the result by Vitali convergence.
\end{proof}

The uniform integrability condition can be satisfied for many algorithms of interest.  We give some examples in the below proposition.

\begin{proposition}
    The sequence $f_d(z) := \mathrm{ESJD}(d,P_{h_d z})$ is uniformly integrable in the following cases:
    \begin{enumerate}
        \item[(i)] $P_{h_d z}$ is the Gaussian random walk Metropolis kernel and $\mu$ has finite mean
        \item[(ii)] $P_{h_d z}$ is the MALA kernel, $\mu$ has finite second moment, and $\pi_d(x) := \prod_{i=1}^d g(x_i)$ with the sequence $\mathbb{E}_{\pi_d}[(\log g)'(X_1)^2]$ uniformly bounded.
    \end{enumerate}
\end{proposition}

\begin{proof}
In the random walk case note that $\mathrm{ESJD}(d, P_{h_d z}) \leq dhz\mathbb{E}[\xi^2] = \ell z$, where $\xi \sim N(0, 1)$. A straightforward dominating function is therefore $u(z) := \ell z$, which is $\mu$-integrable if $\mu$ has finite mean.  In the MALA case the jumping distance satisfies
\[
\mathrm{ESJD}(d, P_{h_d z}) \leq d^{1/3} \left(h_d^2z^2\mathbb{E}[(\log g)'(X_1)^2] + 2h_d z \right) \leq C(z^2 + z)
\]
for large enough $C < \infty$ provided that $\mathbb{E}[(\log g)'(X_1)^2]$ is uniformly bounded. The dominating function $u(z):= C(z^2 + z)$ can therefore be chosen to establish the result provided that $\mu$ has finite second moment.
\end{proof}


Theorem~\ref{thm:scaling_auxvar kernles} implies that, by randomizing the step size, the algorithmic dependence relative to the dimension $d$ stays unchanged.

Moreover, similarly to \eqref{eq:opt_rate_2} and by a change of variable $u = z \ell (\kappa/2)^{2c}$, $\overline \ell_{\rm opt} = \arg \max \{ \overline{\rm eff}(\ell)\}$ now solves
$$
\int_0^\infty\left(u\Phi(-u^{1/(2c)})- \frac{1}{2c} u^{1/(2c) + 1} \phi(-u^{1/(2c)})\right)\mu(u/\omega)d u = 0.
$$
where $\omega = \overline \ell_{\rm opt} (\kappa/2)^{2c}$, yielding new optimal acceptance rates for  Auxiliary-variable Metropolis--Hastings kernels given by 
$$\overline{a}(\overline \ell_{\rm opt})=\int_0^\infty 2\Phi(-(\overline\ell_{\rm opt} z)^{1/(2c)}\kappa/2)\mu(z)d z= \frac{2}{\omega} \int_0^\infty \Phi(-u^{1/(2c)})\mu(u/\omega)d u,
$$
 see Table~\ref{tab:optimal acceptancce rejection} for $\mu$ equal to both Exponential and Uniform distributions and for MALA (with scale $c = 1/3$) and HMC (with scale $c = 1/4$). Note that these values are larger than the optimal acceptance rate of the original MALA and HMC algorithms. 
\begin{table}
  \centering
  \renewcommand{\arraystretch}{1.2}
  \begin{tabular}{|p{2cm}|c|c|c|c|}
    \hline
     & \multicolumn{2}{c|}{MALA ($c=1/3$)} & \multicolumn{2}{c|}{HMC ($c=1/4$)}\\
     \hline
     $\mu$ & $\overline{a}(\overline \ell_{\mathrm{opt}})$ & ${\rm eff} (\ell_{\rm opt})\,/\,\overline{\rm eff}(\overline \ell_{\rm opt})$ & $\overline{a}(\overline \ell_{\mathrm{opt}})$ & ${\rm eff} (\ell_{\rm opt})\,/\,\overline{\rm eff}(\overline \ell_{\rm opt})$\\
     \hline
    Uniform & 0.680  & 1.342 & 0.750 &  1.387\\
    \hline
    Exponential & 0.687 & 1.758 & 0.737 & 1.889\\ \hline
  \end{tabular}
  \caption{Optimal acceptance rate $\overline{a}(\ell_{\mathrm{opt}})$ with $\ell_{\rm opt} =  \arg \max_\ell \overline{\rm eff}(\ell)$ and loss of optimal efficiency $\max({\rm eff}(\ell))\,/\,\max(\overline{\rm eff}(\ell))$ for the Auxiliary-variable MALA and HMC with Uniform and Exponential randomized step sizes.}
\label{tab:optimal acceptancce rejection}
\end{table}

Note that for all $\ell>0$, 
$$
\overline{\rm eff}(\ell) = \int_0^\infty {\rm eff}(\ell z) \mu(z) dz \le  \int_0^\infty {\rm eff}(\ell_{\rm opt}) \mu(z) dz ={\rm eff}(\ell_{\rm opt}) 
$$
where $\ell_{\rm opt} = \arg \max \, {\rm eff}(\ell)$, therefore, the randomized algorithms have a loss of asymptotic efficiency when $\ell$ is optimally tuned.
We quantify this loss in Table~\ref{tab:optimal acceptancce rejection} by comparing the maximum of ${\rm eff}(\ell)$ and $\overline{\rm eff}(\ell)$ for the algorithms considered above.

    \begin{remark}
         Theorem~\ref{thm:scaling_auxvar kernles} explicitly connects scaling limits for Auxiliary-variable kernels $\overline{P_h}$ with those of $P_h$. The same argument cannot be made for Marginalized kernels $M_h$, which are classical Metropolis-Hastings algorithms with a different proposal $\overline Q_h$ as compared to the original $Q_h$.  Scaling limits for the Marginalized kernels must therefore be derived on a case-by-case basis.
    \end{remark}

\section{Numerical experiments}
\label{sec:experiments}

We illustrate our theoretical results and numerically compare our algorithms on standard benchmark target distributions: Neal's funnel \citep{neal2011mcmc} and Rosenbrock's banana distribution \citep{pagani2022n}. We  also numerically show the convergence of our algorithms with initialization in the tails and a Bayesian posterior given by a Poisson regression model. The code for reproducing all the experiments is available at \texttt{https://github.com/SebaGraz/RSS}. 

\subsection{ESJD}
We illustrate first the robustness property of our Auxiliary-variable and Marginalized MALA algorithms by comparing their ESJD with that of standard MALA as the step size $h$ varies. ESJD is estimated with $N = 10^6$ Monte Carlo samples and standard Rao–Blackwellization techniques. Figure~\ref{fig:illustraion} shows these results for different target distributions (Normal, Laplace, and Student-$t$). Note that when $h$ is small the performance is almost identical across all algorithms. For large $h$, however, the ESJD of the randomized algorithms decays more slowly than that of standard MALA, as supported by Theorem~\ref{thm: spectral gap decay}. Furthermore, as indicated by Proposition~\ref{prop:ordering_margvsaux}, Marginalized MALA slightly outperforms Auxiliary-variable MALA in two of the three cases when considering the exponential randomization of the step size. However, for Uniform randomized step sizes, the ESJD of Marginalized and Auxiliary-variable MALA appear indistinguishable.

\begin{figure}[h!]
\centering\includegraphics[width = \textwidth]{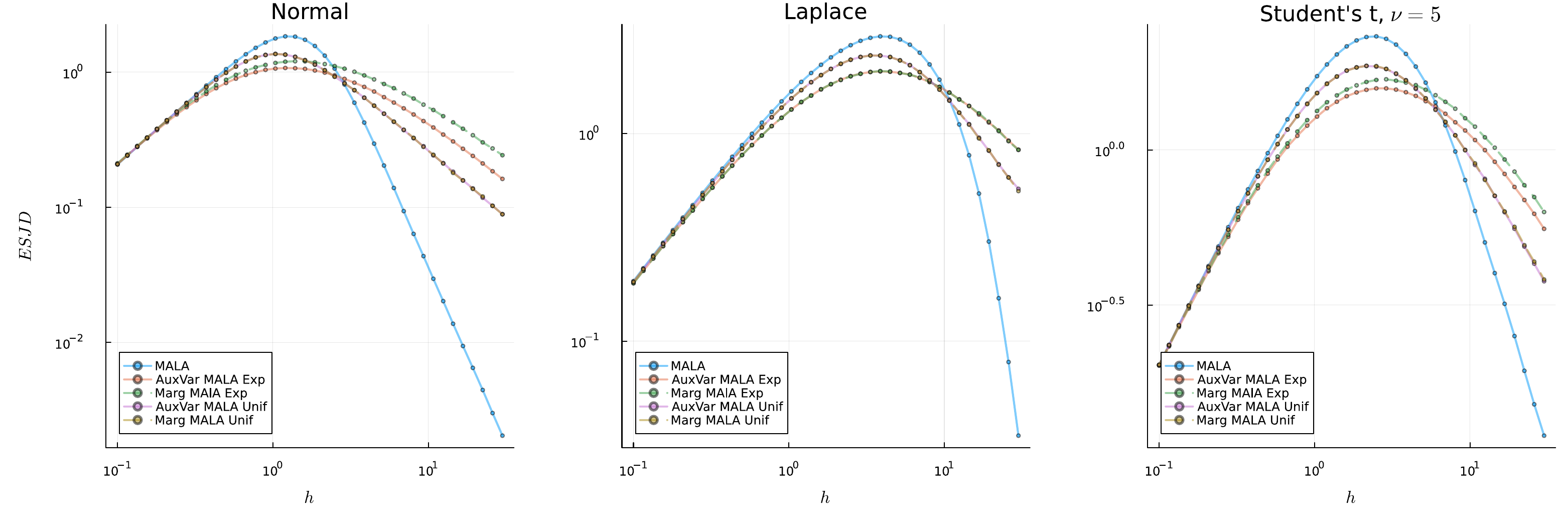}
    \caption{ESJD estimates (log-scale) as a function of the step size ($x$-axis, log-scale) of standard MALA (blue), Auxiliary-variable MALA with Exponential (red) and Uniform (purple) randomized step sizes and Marginalized MALA with Exponential (green) and Uniform (yellow) randomized step sizes, for a one-dimensional standard Normal (left), Laplace (center), Student-$t$ with $5$ degrees of freedom (right). }
    \label{fig:illustraion}
\end{figure}

\subsection{Neal's funnel distribution}\label{sec:neals funnel}
The $d$-dimensional Neal's funnel distribution \cite{neal2011mcmc} is a hierarchical model defined as
\begin{equation}\label{eq:neals_funnel}
X_1 \sim \mathcal{N}(0, \sigma^2), \quad X_i \mid X_1 \overset{\text{i.i.d}}{\sim} \mathcal{N}(0, e^{X_1}), \quad i = 2, \dots, d.
\end{equation}
for some $\sigma^2 >0$. The variable $X_1$ controls the variance of the other variables \( X_{2:d} \). In particular,  when $X_1$ takes large positive values, the variance of the remaining components becomes large and, conversely, for negative values, it becomes small, creating a \emph{funnel-shaped} geometry in the joint distribution, see Figure~\ref{fig:isodensity_targets}, left panel. This structure makes the distribution challenging for sampling algorithms, in particular for popular gradient-based methods such as MALA and HMC, due to the light tails in the $x_1$ direction and irregular geometry. 

\begin{figure}[h!]
    \centering
\includegraphics[width=0.95\linewidth]{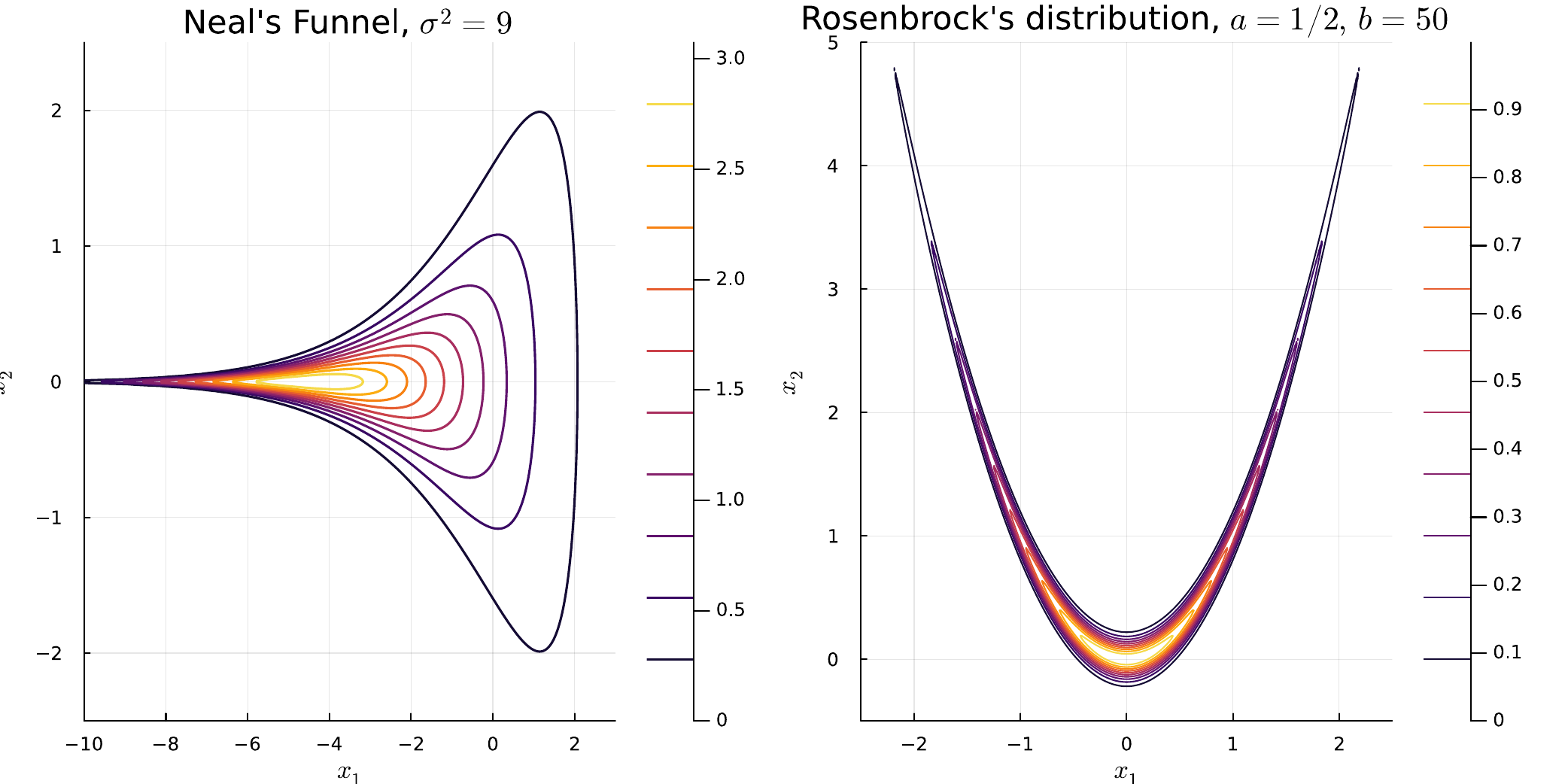}
    \caption{Contour plots of a 2 dimensional Neal's Funnel target with $\sigma^2 = 9$ (left panel) and a 2 dimensional Rosenbrock's banana distribution with $a = 1/2$, $b = 50$ (right panel).}
    \label{fig:isodensity_targets}
\end{figure}

In all experiments, the step size $h_i$  at iteration $i$ is adaptively tuned as in \cite{andrieu2008tutorial}, that is, at iteration $i$, the step size is updated as
\begin{equation}
    \label{eq: step-size adaptation}
    \log h_{i+1} = \log h_{i} + i^{-\beta}(\alpha_i - \alpha^\star)  
\end{equation}
where $\alpha_i$ is the acceptance probability computed at iteration $i$, $\beta = 0.6$ is a learning parameter and  $\alpha^\star \in (0,1)$ is the target probability which is set to be equal to the optimal acceptance rates derived in Section~\ref{subsec:scaling_limits}, Table~\ref{tab:optimal acceptancce rejection}. 
Notice that, during adaptation the results of Sections~\ref{sec:analysis}-\ref{subsec:scaling_limits} do not directly apply. The learning rate $i^{-\beta}$ decreasing quickly to $0$, however, and $h$ stabilizes around a fix value. Then the algorithm behaves as a fixed-$h$ randomized kernel, which is the object studied theoretically.

The left panel of Figure~\ref{fig:FUNNEL} shows the histograms and Q-Q plots of the first component $x_1$ obtained with $N = 10^6$ samples of standard MALA and Auxiliary-variable MALA with Uniform and Exponential randomized step sizes targeting a $d = 10$ Neal's Funnel in \eqref{eq:neals_funnel} with $\sigma^2 = 9$. For all Markov chains we set the same starting value $X_0 \sim \cN(0, I_{d\times d})$, a burn in equal to $N/10$ iterations and an initial step size equal to $1$. Standard MALA fails to explore the neck of the funnel, i.e. the left tail of the first component, but the randomized versions appear to explore this region better. 

The right panel of Figure~\ref{fig:FUNNEL} shows the estimation of $\PP(X_1 < \xi_{0.05})$, where $\xi_{0.05}$ is the 5th percentile, as the number of iterations of the algorithm increases. The box-plots are obtained with $20$ independent runs of the Markov chains targeting \eqref{eq:neals_funnel} with $\sigma^2 = 4$. Note that, for standard MALA, the upper and lower quantiles are equal to 0 in all experiments, meaning that the majority of the MCMC runs failed to  reach values below $\xi_{0.05}$ on their first component. 
\begin{figure}[h!]
    \centering
    \includegraphics[width=0.48\linewidth]{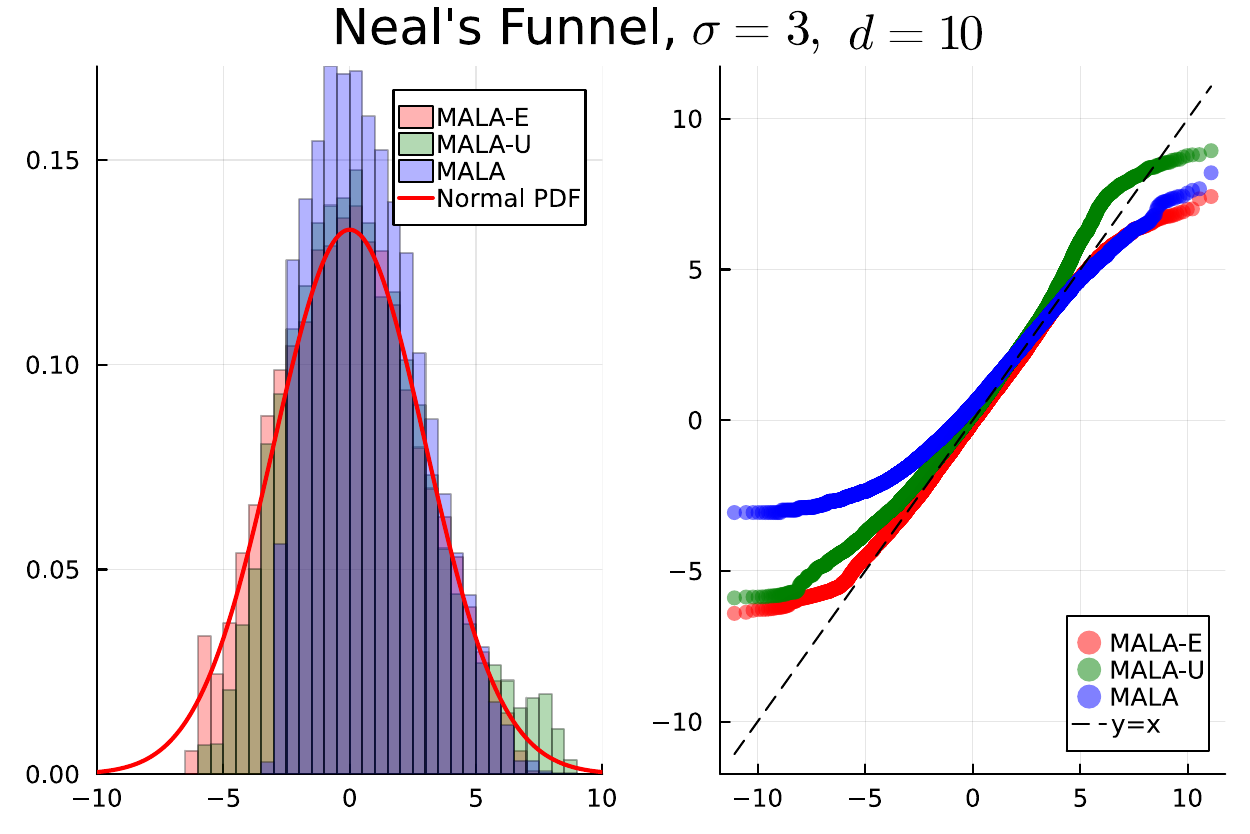}
        \includegraphics[width=0.48\linewidth]{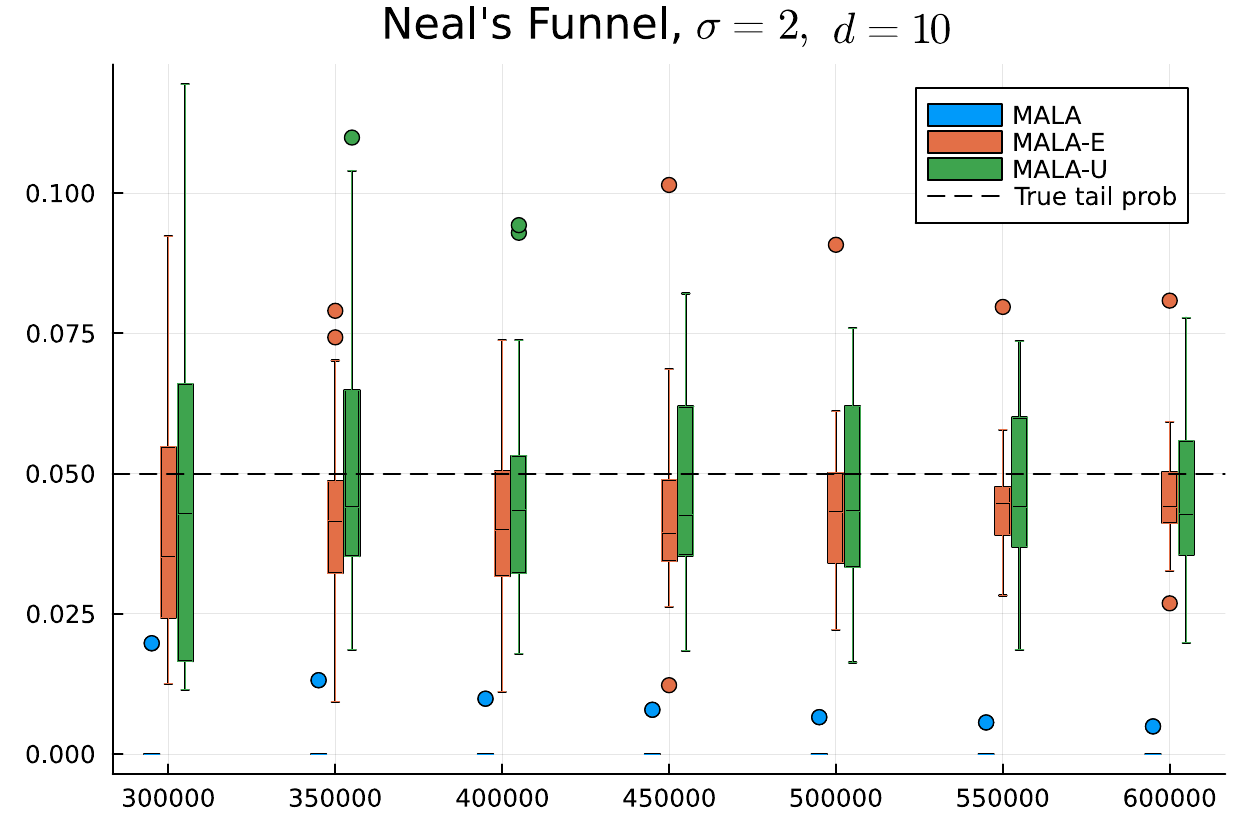}
    \caption{Left panel: histograms and Q-Q plots of the first components $x_1$ for MALA (blue) and Auxiliary-variable MALA with Uniform (green) and Exponential (red) randomized step sizes. The red solid line shows the true marginal density of $X_1$. Right panel: box plots for the estimation $\PP(X_1 < \xi_{0.05})$ over $20$ independent simulations of each algorithm, as the number of iterations ($x$-axis) increases.}
    \label{fig:FUNNEL}
\end{figure}

\subsection{The Rosenbrock distribution}
The two-dimensional Rosenbrock  distribution is defined as
\begin{equation}
    \label{eq: rosenbrock}
    X_1 \sim \mathcal N\!\left(0,\tfrac{1}{2a}\right),
\qquad
X_{2}\mid X_{1} \sim \mathcal N\!\left(X_{1}^2,\tfrac{1}{2b}\right)
\end{equation}

for parameters $a>0$ and $b>0$.
The nonlinear dependence between $X_1$ and $X_2$ produces a curved, banana-shaped density, see Figure~\ref{fig:isodensity_targets}, right panel.

We perform similar experiments as in  Section~\ref{sec:neals funnel}, where the step size $h$ is adaptively tuned as in \eqref{eq: step-size adaptation}. Figure~\ref{fig:BANANA}, left panel, shows the histograms and Q-Q plots of the first component obtained with $N = 2 \times 10^6$ samples of MALA and Auxiliary-variable MALA with Exponential and Uniform randomized step sizes targeting \eqref{eq: rosenbrock} with $a = 1/2, \, b = 50$. Initialization and burn in are set as in Section~\ref{sec:neals funnel}. The right-hand side shows the estimation of $\PP(|X_1|>\xi_{0.975})$ for the same target, where $\xi_{0.975}$ is the 975th percentile, as the number of iterations increases. The box plots are again obtained with 20 independent runs of each algorithm. 

The results suggest that the Auxiliary-variable MALA algorithms improve the exploration of the tails, although the differences in performance with respect to standard MALA are less dramatic compared with Section~\ref{sec:neals funnel}.
\begin{figure}[h!]
    \centering
    \includegraphics[width=0.48\linewidth]{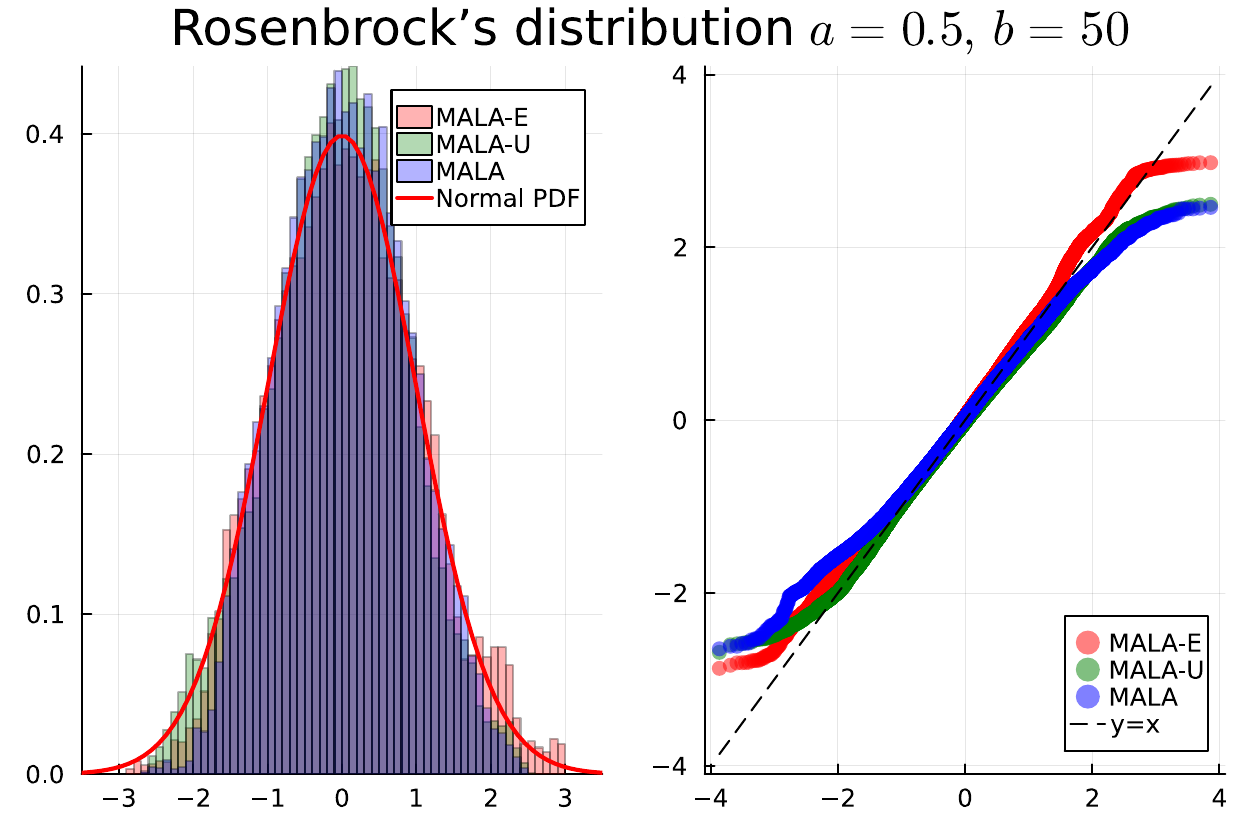}
    \includegraphics[width=0.48\linewidth]{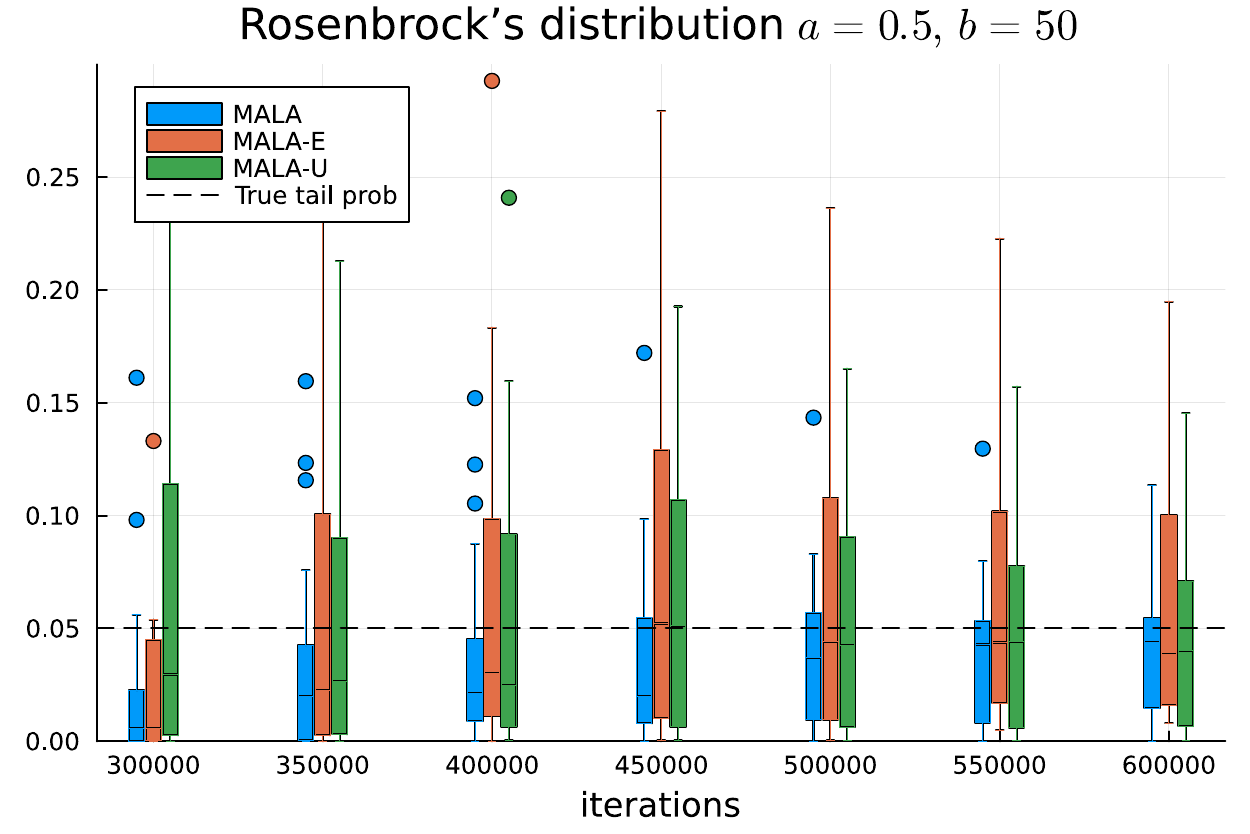}
    \caption{Left panel: same as in Figure~\ref{fig:FUNNEL}.  Right panel: box plots for the estimation $\PP(|X_1| > \xi_{0.975})$ over $20$ independent simulations of each algorithm, as the number of iterations increases ($x$-axis).}
    \label{fig:BANANA}
\end{figure}

\subsection{Poisson regression model}
For $i =1,2,\dots,n$, consider a regression model given by 
\begin{equation}
\label{eq: poisson_model}
Y_i \mid z_i \overset{\text{i.i.d}}{\sim}\mathrm{Poiss}(\exp(\langle z_i, x \rangle )),
\end{equation}
with unknown parameter $x \in \RR^d$ and sample size $n = 10d$. We set $d = 50$, $x \sim \cN(0,I_{d\times d})$ and we simulate synthetic covariates $z_i \sim \cN(0, I_{d\times d}/d)$ and $Y_i$ from \eqref{eq: poisson_model} for $i = 1,2,\dots,n$. The posterior distribution $\pi$, given by the product of the likelihood and a  standard Gaussian prior on $x$, is such that $\log \pi$ decays exponentially fast in the tails, creating a challenging target in particular at the tails of the distribution.  Figure~\ref{fig:poisson_reg} shows the traces of MALA and our Auxiliary-variable algorithms, with adaptive step size as in \eqref{eq: step-size adaptation}. All algorithms were initialized in the tail of the distribution at $X_0 \sim \cN(0, 10^2 \cdot I_{d\times d})$. The  Auxiliary-variable algorithms appear to reach the bulk of the distribution faster than standard MALA.

\begin{figure}
    \centering
    \includegraphics[width=0.9\linewidth]{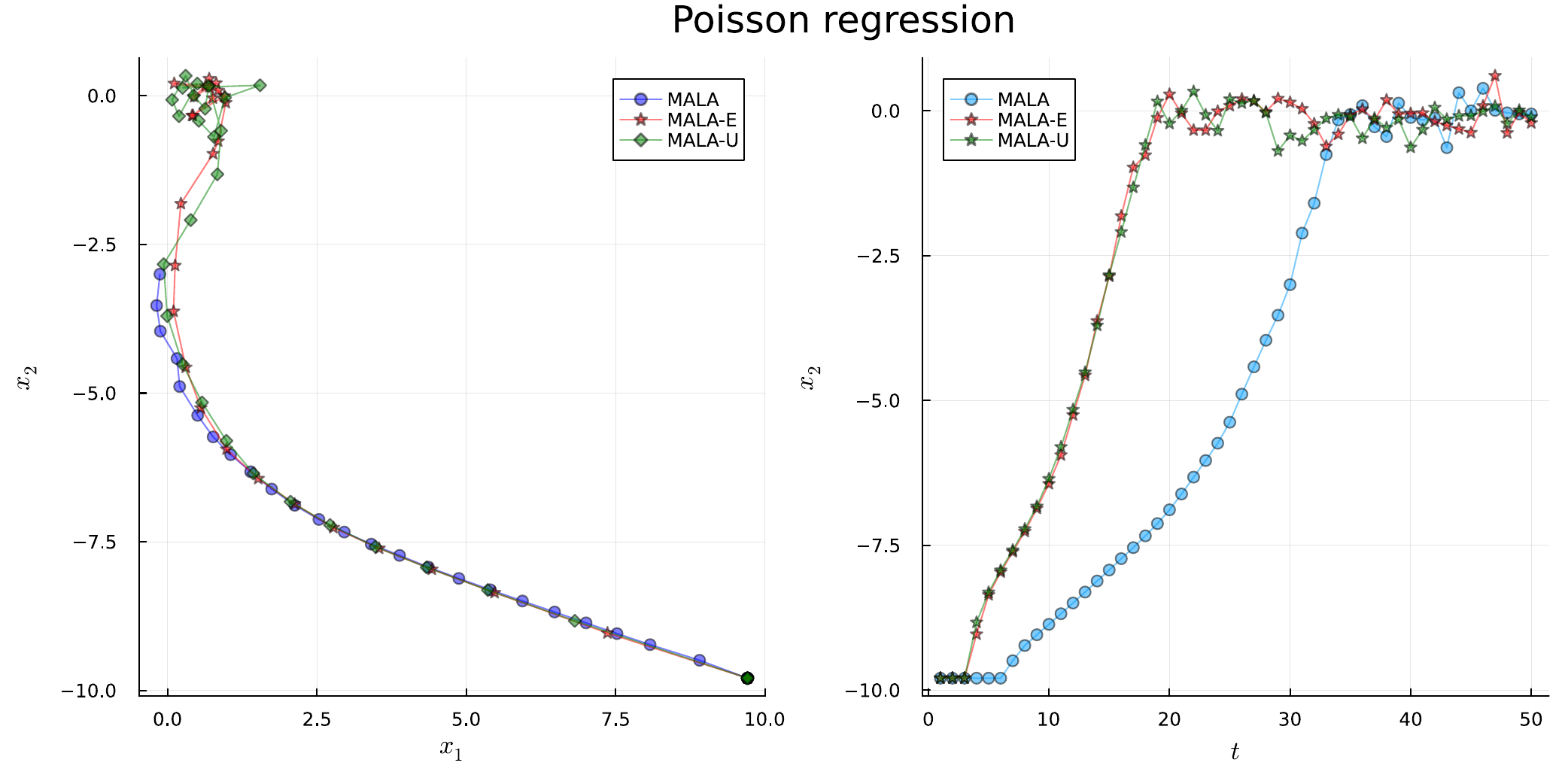}
    \caption{Convergence of MALA (blue) and Auxiliary-variable MALA with Uniform (green) and Exponential (red) randomized step sizes initialized in the tails of the distribution. Left panel: trajectories of the first two coordinates. Right panel: traces of the first coordinate.}
    \label{fig:poisson_reg}
\end{figure}
\section{Discussion}

We have shown how the randomization of a single one-dimensional parameter within a Metropolis--Hastings kernel can bring surprising benefits in terms of practical performance, and provided some novel theoretical results to explain this phenomenon.  A natural extension is to consider the randomization of vector- or matrix-valued parameters.  One obvious example is to randomize either a diagonal or dense preconditioning matrix.  The recent work of \cite{hird2025quantifying} highlights scenarios in which spectral gaps and mixing times can be improved through the careful choice of a fixed matrix to precondition sampling algorithms, which is a popular element of the applied sampling toolkit, but little consideration has been given to the idea of randomizing this matrix at each iteration of the algorithm by sampling from some appropriate distribution.  Indeed, in the case that this matrix is diagonal, the multi-dimensional Gaussian random walk Metropolis and Barker proposal algorithms can again be viewed as marginalized versions of simpler base kernels analogous to the one-dimensional scenario described in Section \ref{subsec:examples_randomized_kernels}. Another natural question concerning kernels with a randomized step size such as $\overline P_h$ and $M_h$ is the optimal choice of $\mu$, the randomizing distribution.  Related questions concerning the optimal choice of proposal noise in the random walk Metropolis and Barker proposal have been considered previously in \cite{vogrinc2023optimal, neal2011optimal}, but there are many ways in which such a question can be tackled and undoubtedly more to be learned.  We leave a thorough exploration of these topics for future work.

We note in passing that a version of Hamiltonian Monte Carlo with a modified kinetic energy, as studied in \cite{Livingstone::2019, zhang2016laplacian}, combined with a randomized step size, produces a marginalized algorithm which is very close to the Barker proposal scheme of \citet{livingstone2022barker} in one dimension.  In the case that the kinetic energy is $|p|$, meaning the augmented momentum variable $p$ follows a Laplace distribution, then a Hamiltonian proposal with a single leapfrog step takes the form
\[
x' = x + h \sign \left( p_0 + \frac{h}{2}(\log\pi)'(x) \right).
\]
If $h$ is then randomized using a half-Normal distribution, straightforward calculations show that the associated proposal density $Q_h$ becomes
\[
Q_h(x,y) = 2 F_L\left( (\log\pi)'(x)\frac{y-x}{h}\right) \mu\left(\frac{y-x}{h}\right),
\]
where $F_L$ denotes the cumulative distribution function (CDF) of the Laplace distribution. This represents a skew-symmetric distribution as described in \cite{azzalini2013skew}, and differs from the Barker proposal only in the choice of CDF, which is instead taken from the logistic distribution in the Barker case, for reasons outlined in \cite{hird2020fresh, livingstone2025foundations}.  A version of this algorithm with multiple leapfrog steps would therefore be of interest, in combining the robustness of Barker and of the randomized kernels with the advantages of methods with auxiliary momentum like Hamiltonian Monte Carlo. Again we leave a detailed exploration for future work.

\section*{Acknowledgement}
SG acknowledges support from the European Research Council (ERC), through the Starting grant ‘PrSc-HDBayLe’, project number 101076564. SL thanks Karan Garg for preliminary numerical work related to Section \ref{sec:experiments} during his MSc project at UCL.

\section*{Conflict of interest}
The authors declare that they have no conflict of interest
\bibliography{reference} 

@article{pagani2022n,
  title={An n-dimensional Rosenbrock distribution for Markov chain Monte Carlo testing},
  author={Pagani, Filippo and Wiegand, Martin and Nadarajah, Saralees},
  journal={Scandinavian Journal of Statistics},
  volume={49},
  number={2},
  pages={657--680},
  year={2022},
  publisher={Wiley Online Library}
}

@inproceedings{hird2020fresh,
  title={A fresh take on ‘Barker dynamics’ for MCMC},
  author={Hird, Max and Livingstone, Samuel and Zanella, Giacomo},
  booktitle={International Conference on Monte Carlo and Quasi-Monte Carlo Methods in Scientific Computing},
  pages={169--184},
  year={2020},
  organization={Springer}
}

@article{andrieu2008tutorial,
  title={A tutorial on adaptive MCMC},
  author={Andrieu, Christophe and Thoms, Johannes},
  journal={Statistics and computing},
  volume={18},
  pages={343--373},
  year={2008},
  publisher={Springer}
}

@article{yang2020optimal,
  title={Optimal scaling of random-walk metropolis algorithms on general target distributions},
  author={Yang, Jun and Roberts, Gareth O and Rosenthal, Jeffrey S},
  journal={Stochastic Processes and their Applications},
  volume={130},
  number={10},
  pages={6094--6132},
  year={2020},
  publisher={Elsevier}
}

@article{storvik2011flexibility,
  title={On the flexibility of Metropolis--Hastings acceptance probabilities in auxiliary variable proposal generation},
  author={Storvik, Geir},
  journal={Scandinavian Journal of Statistics},
  volume={38},
  number={2},
  pages={342--358},
  year={2011},
  publisher={Wiley Online Library}
}

@article{titsias2018auxiliary,
  title={Auxiliary gradient-based sampling algorithms},
  author={Titsias, Michalis K and Papaspiliopoulos, Omiros},
  journal={Journal of the Royal Statistical Society Series B: Statistical Methodology},
  volume={80},
  number={4},
  pages={749--767},
  year={2018},
  publisher={Oxford University Press}
}

@article{roberts2004general,
    AUTHOR = {Roberts, Gareth O. and Rosenthal, Jeffrey S.},
     TITLE = {General state space {M}arkov chains and {MCMC} algorithms},
   JOURNAL = {Probab. Surv.},
  FJOURNAL = {Probability Surveys},
    VOLUME = {1},
      YEAR = {2004},
     PAGES = {20--71},
   MRCLASS = {60J05 (62F10 65C05)},
  MRNUMBER = {2095565},
MRREVIEWER = {Roman Urban},
       DOI = {10.1214/154957804100000024},
       URL = {https://doi.org/10.1214/154957804100000024},
}

@article{durmus2020irreducibility,
    AUTHOR = {Durmus, Alain and Moulines, \'{E}ric and Saksman, Eero},
     TITLE = {Irreducibility and geometric ergodicity of {H}amiltonian
              {M}onte {C}arlo},
   JOURNAL = {Ann. Statist.},
  FJOURNAL = {The Annals of Statistics},
    VOLUME = {48},
      YEAR = {2020},
    NUMBER = {6},
     PAGES = {3545--3564},
      ISSN = {0090-5364},
   MRCLASS = {60J05 (60J22 65C05)},
  MRNUMBER = {4185819},
MRREVIEWER = {Neng-Yi Wang},
       DOI = {10.1214/19-AOS1941},
       URL = {https://doi.org/10.1214/19-AOS1941},
}

@article{livingstone2019geometric,
    AUTHOR = {Livingstone, Samuel and Betancourt, Michael and Byrne, Simon
              and Girolami, Mark},
     TITLE = {On the geometric ergodicity of {H}amiltonian {M}onte {C}arlo},
   JOURNAL = {Bernoulli},
  FJOURNAL = {Bernoulli. Official Journal of the Bernoulli Society for
              Mathematical Statistics and Probability},
    VOLUME = {25},
      YEAR = {2019},
    NUMBER = {4A},
     PAGES = {3109--3138},
      ISSN = {1350-7265},
   MRCLASS = {60J22 (65C05)},
  MRNUMBER = {4003576},
MRREVIEWER = {Bernd Heidergott},
       DOI = {10.3150/18-BEJ1083},
       URL = {https://doi.org/10.3150/18-BEJ1083},
}

@article{andrieu2018uniform,
  title={Uniform ergodicity of the iterated conditional SMC and geometric ergodicity of particle Gibbs samplers},
  author={Andrieu, Christophe and Lee, Anthony and Vihola, Matti},
  journal={Bernoulli},
  pages={842--872},
  year={2018},
  publisher={JSTOR}
}

@article{vogrinc2023optimal,
  title={Optimal design of the Barker proposal and other locally balanced Metropolis--Hastings algorithms},
  author={Vogrinc, Jure and Livingstone, Samuel and Zanella, Giacomo},
  journal={Biometrika},
  volume={110},
  number={3},
  pages={579--595},
  year={2023},
  publisher={Oxford University Press}
}

@article{rosenthal2003asymptotic,
  title={Asymptotic variance and convergence rates of nearly-periodic Markov chain Monte Carlo algorithms},
  author={Rosenthal, Jeffrey S},
  journal={Journal of the American Statistical Association},
  volume={98},
  number={461},
  pages={169--177},
  year={2003},
  publisher={Taylor \& Francis}
}

@article{andrieu2021peskun,
    AUTHOR = {Andrieu, Christophe and Livingstone, Samuel},
     TITLE = {Peskun-{T}ierney ordering for {M}arkovian {M}onte {C}arlo:
              beyond the reversible scenario},
   JOURNAL = {Ann. Statist.},
  FJOURNAL = {The Annals of Statistics},
    VOLUME = {49},
      YEAR = {2021},
    NUMBER = {4},
     PAGES = {1958--1981},
      ISSN = {0090-5364},
   MRCLASS = {65C05 (60J05 60J22 62J10)},
  MRNUMBER = {4319237},
MRREVIEWER = {Pavel T. Stoynov},
       DOI = {10.1214/20-aos2008},
       URL = {https://doi.org/10.1214/20-aos2008},
}

@article{metropolis1953equation,
  title={Equation of state calculations by fast computing machines},
  author={Metropolis, Nicholas and Rosenbluth, Arianna W and Rosenbluth, Marshall N and Teller, Augusta H and Teller, Edward},
  journal={The journal of chemical physics},
  volume={21},
  number={6},
  pages={1087--1092},
  year={1953},
  publisher={American Institute of Physics}
}

@article{livingstone2022barker,
  title={The Barker proposal: Combining robustness and efficiency in gradient-based MCMC},
  author={Livingstone, Samuel and Zanella, Giacomo},
  journal={Journal of the Royal Statistical Society Series B: Statistical Methodology},
  volume={84},
  number={2},
  pages={496--523},
  year={2022},
  publisher={Oxford University Press}
}

@Article{Beskos::2013,
    AUTHOR = {Beskos, Alexandros and Pillai, Natesh and Roberts, Gareth and
              Sanz-Serna, Jesus-Maria and Stuart, Andrew},
     TITLE = {Optimal tuning of the hybrid {M}onte {C}arlo algorithm},
   JOURNAL = {Bernoulli},
  FJOURNAL = {Bernoulli. Official Journal of the Bernoulli Society for
              Mathematical Statistics and Probability},
    VOLUME = {19},
      YEAR = {2013},
    NUMBER = {5A},
     PAGES = {1501--1534},
      ISSN = {1350-7265},
   MRCLASS = {60J22 (65C05)},
  MRNUMBER = {3129023},
       DOI = {10.3150/12-BEJ414},
       URL = {https://doi.org/10.3150/12-BEJ414},
}

@Article{Livingstone::2019,
    AUTHOR = {Livingstone, S. and Faulkner, M. F. and Roberts, G. O.},
     TITLE = {Kinetic energy choice in {H}amiltonian/hybrid {M}onte {C}arlo},
   JOURNAL = {Biometrika},
  FJOURNAL = {Biometrika},
    VOLUME = {106},
      YEAR = {2019},
    NUMBER = {2},
     PAGES = {303--319},
      ISSN = {0006-3444},
   MRCLASS = {60J22 (60J05 62F15 65C05)},
  MRNUMBER = {3949305},
       DOI = {10.1093/biomet/asz013},
       URL = {https://doi.org/10.1093/biomet/asz013},
}

@Article{Roberts::2001,
    AUTHOR = {Roberts, Gareth O. and Rosenthal, Jeffrey S.},
     TITLE = {Optimal scaling for various {M}etropolis-{H}astings
              algorithms},
   JOURNAL = {Statist. Sci.},
  FJOURNAL = {Statistical Science. A Review Journal of the Institute of
              Mathematical Statistics},
    VOLUME = {16},
      YEAR = {2001},
    NUMBER = {4},
     PAGES = {351--367},
      ISSN = {0883-4237},
   MRCLASS = {65C05 (60J05)},
  MRNUMBER = {1888450},
       DOI = {10.1214/ss/1015346320},
       URL = {https://doi.org/10.1214/ss/1015346320},
}

@Article{Roberts::1997,
    AUTHOR = {Roberts, G. O. and Gelman, A. and Gilks, W. R.},
     TITLE = {Weak convergence and optimal scaling of random walk
              {M}etropolis algorithms},
   JOURNAL = {Ann. Appl. Probab.},
  FJOURNAL = {The Annals of Applied Probability},
    VOLUME = {7},
      YEAR = {1997},
    NUMBER = {1},
     PAGES = {110--120},
      ISSN = {1050-5164},
   MRCLASS = {60F05 (60J10 60J60 65U05)},
  MRNUMBER = {1428751},
MRREVIEWER = {Jean Diebolt},
       DOI = {10.1214/aoap/1034625254},
       URL = {https://doi.org/10.1214/aoap/1034625254},
}

@Article{Roberts::1998,
    AUTHOR = {Roberts, Gareth O. and Rosenthal, Jeffrey S.},
     TITLE = {Optimal scaling of discrete approximations to {L}angevin
              diffusions},
   JOURNAL = {J. R. Stat. Soc. Ser. B Stat. Methodol.},
  FJOURNAL = {Journal of the Royal Statistical Society. Series B.
              Statistical Methodology},
    VOLUME = {60},
      YEAR = {1998},
    NUMBER = {1},
     PAGES = {255--268},
      ISSN = {1369-7412},
   MRCLASS = {60J05 (65C05)},
  MRNUMBER = {1625691},
MRREVIEWER = {Jean Diebolt},
       DOI = {10.1111/1467-9868.00123},
       URL = {https://doi.org/10.1111/1467-9868.00123},
}

@Article{Hastings::1970,
    AUTHOR = {Hastings, W. K.},
     TITLE = {Monte {C}arlo sampling methods using {M}arkov chains and their
              applications},
   JOURNAL = {Biometrika},
  FJOURNAL = {Biometrika},
    VOLUME = {57},
      YEAR = {1970},
    NUMBER = {1},
     PAGES = {97--109},
      ISSN = {0006-3444},
   MRCLASS = {65C05 (11K45 60J22)},
  MRNUMBER = {3363437},
       DOI = {10.1093/biomet/57.1.97},
       URL = {https://doi.org/10.1093/biomet/57.1.97},
}

@Article{RobertsTweedie::1996,
    AUTHOR = {Roberts, Gareth O. and Tweedie, Richard L.},
     TITLE = {Exponential convergence of {L}angevin distributions and their
              discrete approximations},
   JOURNAL = {Bernoulli},
  FJOURNAL = {Bernoulli. Official Journal of the Bernoulli Society for
              Mathematical Statistics and Probability},
    VOLUME = {2},
      YEAR = {1996},
    NUMBER = {4},
     PAGES = {341--363},
      ISSN = {1350-7265},
   MRCLASS = {62E25 (65C05)},
  MRNUMBER = {1440273},
MRREVIEWER = {Arnoldo Frigessi},
       DOI = {10.2307/3318418},
       URL = {https://doi.org/10.2307/3318418},
}

@article{rudolf2011explicit,
  title={Explicit error bounds for Markov chain Monte Carlo},
  author={Rudolf, Daniel},
  journal={arXiv preprint arXiv:1108.3201},
  year={2011}
}

@article{hofstadler2025optimal,
  title={Optimal convergence rates of MCMC integration for functions with unbounded second moment},
  author={Hofstadler, Julian},
  journal={Journal of Applied Probability},
  volume={62},
  number={3},
  pages={1069--1075},
  year={2025},
  publisher={Cambridge University Press}
}

@article{riou2022metropolis,
  title={Metropolis adjusted Langevin trajectories: a robust alternative to Hamiltonian Monte Carlo},
  author={Riou-Durand, Lionel and Vogrinc, Jure},
  journal={arXiv preprint arXiv:2202.13230},
  year={2022}
}

@article{roberts2016complexity,
  title={Complexity bounds for Markov chain Monte Carlo algorithms via diffusion limits},
  author={Roberts, Gareth O and Rosenthal, Jeffrey S},
  journal={Journal of Applied Probability},
  volume={53},
  number={2},
  pages={410--420},
  year={2016},
  publisher={Cambridge University Press}
}

@incollection{neal2011mcmc,
  title        = {{MCMC} using {H}amiltonian dynamics},
  author       = {Neal, Radford M.},
  booktitle    = {Handbook of Markov Chain Monte Carlo},
  editor       = {Brooks, Steve and Gelman, Andrew and Jones, Galin L. and Meng, Xiao-Li},
  publisher    = {Chapman and Hall/CRC},
  year         = {2011},
  chapter      = {5},
  pages        = {113--162},
}

@article{andrieu2022comparison,
  title={Comparison of Markov chains via weak Poincar{\'e} inequalities with application to pseudo-marginal MCMC},
  author={Andrieu, Christophe and Lee, Anthony and Power, Sam and Wang, Andi Q},
  journal={The Annals of Statistics},
  volume={50},
  number={6},
  pages={3592--3618},
  year={2022},
  publisher={Institute of Mathematical Statistics}
}

@techreport{mira1999ordering,
  title={Ordering Monte Carlo Markov chains},
  author={Mira, Antonietta and Geyer, Charles J},
  year={1999},
  institution={University of Minnesota}
}

@article{tierney1998note,
  title={A note on Metropolis-Hastings kernels for general state spaces},
  author={Tierney, Luke},
  journal={Annals of applied probability},
  pages={1--9},
  year={1998},
  publisher={JSTOR}
}

@article{peskun1973optimum,
  title={Optimum monte-carlo sampling using markov chains},
  author={Peskun, Peter H},
  journal={Biometrika},
  volume={60},
  number={3},
  pages={607--612},
  year={1973},
  publisher={Oxford University Press}
}

@article{iguchi2024skew,
  title={Skew-symmetric schemes for stochastic differential equations with non-Lipschitz drift: an unadjusted Barker algorithm},
  author={Iguchi, Yuga and Livingstone, Samuel and N{\"u}sken, Nikolas and Vasdekis, Giorgos and Zhang, Rui-Yang},
  journal={arXiv preprint arXiv:2405.14373},
  year={2024}
}

@article{sherlock2009optimal,
  title={Optimal scaling of the random walk Metropolis on elliptically symmetric unimodal targets},
  author={Sherlock, Chris and Roberts, Gareth},
  journal={Bernoulli},
  pages={774--798},
  year={2009}
}

@inproceedings{zhang2016laplacian,
  title={Laplacian Hamiltonian Monte Carlo},
  author={Zhang, Yizhe and Chen, Changyou and Henao, Ricardo and Carin, Lawrence},
  booktitle={Joint European Conference on Machine Learning and Knowledge Discovery in Databases},
  pages={98--114},
  year={2016},
  organization={Springer}
}

@article{livingstone2025foundations,
  title={Foundations of locally-balanced Markov processes},
  author={Livingstone, Samuel and Vasdekis, Giorgos and Zanella, Giacomo},
  journal={arXiv preprint arXiv:2504.13322},
  year={2025}
}

@book{azzalini2013skew,
  title={The skew-normal and related families},
  author={Azzalini, Adelchi},
  volume={3},
  year={2013},
  publisher={Cambridge University Press}
}

@article{hird2025quantifying,
  title={Quantifying the effectiveness of linear preconditioning in Markov chain Monte Carlo},
  author={Hird, Max and Livingstone, Samuel},
  journal={Journal of Machine Learning Research},
  volume={26},
  number={119},
  pages={1--51},
  year={2025}
}

@article{neal2011optimal,
  title={Optimal scaling of random walk Metropolis algorithms with non-Gaussian proposals},
  author={Neal, Peter and Roberts, Gareth},
  journal={Methodology and Computing in Applied Probability},
  volume={13},
  number={3},
  pages={583--601},
  year={2011},
  publisher={Springer}
}

@book{jorgensen2012statistical,
  title={Statistical properties of the generalized inverse Gaussian distribution},
  author={Jorgensen, Bent},
  volume={9},
  year={2012},
  publisher={Springer Science \& Business Media}
}
\end{document}